\documentclass[11pt]{article}
\usepackage{amssymb}
\usepackage{amsfonts}
\usepackage{amsmath}
\usepackage[numbers]{natbib}
\usepackage[nohead]{geometry}
\usepackage[singlespacing]{setspace}
\usepackage[bottom]{footmisc}
\usepackage{indentfirst}
\usepackage{endnotes}
\usepackage{graphicx}%
\usepackage{subcaption}
\usepackage{rotating}
\usepackage{amsthm}
\usepackage{booktabs}
\usepackage{tikz}
\usepackage{algorithm}
\usepackage{amsfonts,amsmath,amssymb,amsthm}
\usepackage{comment}
\usepackage{bbding}
\usepackage{xr}
\graphicspath{{Figure/}}
\DeclareGraphicsExtensions{.png, .pdf, .jpg}

\newcommand{\be} {\begin{eqnarray*}}
\newcommand{\ee} {\end{eqnarray*}}

\newcommand{\argmin}{\mathop{\rm argmin~}}

\newtheorem{theorem}{Theorem}[section]
\newtheorem{lemma}[theorem]{Lemma}
\newtheorem{ass}[theorem]{Assumption}

\newtheorem{rem}{Remark}[section]
\newtheorem{example}{Example}[section]

 
\begin{document}
\title{\ Covariate Adaptive False Discovery Rate Control with Applications to Omics-Wide Multiple Testing}
\author{Xianyang Zhang and Jun Chen}
\date{}
\footnotetext[1]{Xianyang Zhang (zhangxiany@stat.tamu.edu) is Associate
Professor of Statistics at Texas A\&M University. Jun Chen (Chen.Jun2@mayo.edu) is Associate Professor of Biostatistics at Mayo Clinic. Zhang acknowledges partial support from
NSF DMS-1830392 and NSF DMS-1811747. Chen acknowledges support from Mayo Clinic Center for Individualized Medicine.} \maketitle
\onehalfspacing
\maketitle
\sloppy%

\textbf{Abstract} Conventional multiple testing procedures often assume hypotheses for different features are exchangeable. However, in many scientific applications, additional covariate
information regarding the patterns of signals and nulls are available. In this paper, we introduce an FDR control procedure in large-scale inference problem that can incorporate covariate information. We develop a fast algorithm to implement the proposed procedure and prove its asymptotic validity even when the underlying likelihood ratio model is misspecified and the p-values are weakly dependent (e.g., strong mixing). Extensive simulations are conducted to study the finite sample performance of the proposed method and we demonstrate that the new approach improves over the state-of-the-art approaches { by being  flexible, robust, powerful and computationally efficient}. We finally apply the method to several omics datasets arising from genomics studies with the aim to identify omics features associated with some clinical and biological phenotypes. We show that the method is overall the most powerful among competing methods, especially when the signal is sparse. The proposed \textbf{C}ovariate \textbf{A}daptive \textbf{M}ultiple \textbf{T}esting procedure is implemented in the R
package \texttt{CAMT}.
\\
\strut \textbf{Keywords:} Covariates, EM-algorithm, False Discovery Rate, Multiple Testing.

\doublespacing

\section{Introduction}

Multiple testing refers to simultaneous testing of more than one
hypothesis. Given a set of hypotheses, multiple testing deals with
deciding which hypotheses to reject while guaranteeing some notion
of control on the number of false rejections. A traditional measure
is the family-wise error rate (FWER), which is the probability of
committing at least one type I error. As the number of trials
increases, FWER still measures the probability of at least one false
discovery, which is overly stringent in many applications. This
absolute control is in contrast to the proportionate control
afforded by the false discovery rate (FDR).

Consider the problem of testing $m$ distinct hypotheses. Suppose a
multiple testing procedure rejects $R$ hypotheses among which $V$
hypotheses are null, i.e., it commits $V$ type I errors. In the
seminal paper by Benjamini and Hochberg, the authors introduced the
concept of FDR defined as
$$\text{FDR}=E\left[\frac{V}{R\vee
1}\right],$$ where $a\vee b=\max\{a,b\}$ for $a,b\in\mathbb{R}$, and the expectation is with respect to the random quantities $V$ and $R$.
FDR has many advantageous features comparing to other existing error
measures. Control of FDR is less stringent than the control of FWER
especially when a large number of hypothesis tests are performed. FDR is also
adaptive to the underlying signal structure in the data. The
widespread use of FDR is believed to stem from and motivated by the
modern technologies which produce big datasets, with huge
numbers of measurements on a comparatively small number of
experimental units. Another reason for the popularity of FDR is the
existence of a simple and easy-to-use procedure proposed in
Benjamini and Hochberg (1995) (the BH procedure, hereafter) to
control the FDR at a prespecified level.

Although the BH procedure is more powerful than procedures aiming to
control the FWER, it assumes hypotheses for different features
are exchangeable which could result in suboptimal power as demonstrated in recent literature when individual tests
differ in their true effect size, signal-to-noise ratio or prior probability of being false.
In many scientific applications, particularly those from genomics studies,
there are rich covariates that are informative of either the statistical power or the prior null probability.
These covariates can be roughly derived into two classes: statistical covariates  and external covariates (Ignatiadi et al., 2016).  Statistical covariates are derived from the data itself and could reflect the power or null probability. Generic statistical covariates include the sample variance, total sample size and sample size ratio (for two-group comparison), and the direction of the effects. There are also specific statistical covariates for particular applications. For example, in transcriptomics studies using RNA-Seq,  the sum of read counts per gene across all
samples is a statistical covariate informative of power since the low-count genes are subject to more sampling variability.   Similarly, the minor allele frequency and the prevalence of the bacterial species can be taken as statistical covariates for genome-wide association studies (GWAS) and microbiome-wide association studies (MWAS), respectively.   Moreover, the average methylation level of a CpG site in epigenome-wide association studies (EWAS) can be a statistical covariate informative of the prior null probability  due to the fact that differential methylation frequently occurs in highly or lowly methylated region depending on the biological context.  Besides these statistical covariates, there are a plethora of covariates that are derived from external sources and  are usually informative of the prior null probability.     These external covariates include the deleteriousness of the genetic variants for GWAS, the location (island and shore) of CpG methylation variants for EWAS,   and pathogenicity  of the bacterial species  for MWAS.  Useful external covariates also include p-values from previous or related studies which suggest that some hypotheses are more likely to be non-null than others. Exploiting such external covariates in multiple testing could lead to improved statistical power as well as
enhanced interpretability of research results.


Accommodating covariates in multiple testing has recently been a very active
research area. We briefly review some contributions that are most
relevant to the current work. The basic idea of many existing works
is to relax the p-value thresholds for hypotheses that are more
likely to be non-null and tighten the thresholds for the other
hypotheses so that the overall FDR level can be controlled. For
example, Genovese et al. (2006) proposed to weight the p-values with
different weights, and then apply the BH procedure to the weighted
p-values. Hu et al. (2010) developed a group BH procedure by
estimating the proportions of null hypotheses for each group
separately, which extends the method in Storey (2002). Li and Barber
(2017) generalized this idea by using the censored p-values (i.e.,
p-values that are greater than a pre-specified threshold) to
adaptively estimate the weights that can be designed to reflect any
structure believed to be present. Ignatiadi et al. (2016) proposed
the independent hypothesis weighting (IHW) for multiple testing with
covariate information. Their idea is to bin the covariates into
several groups and then apply the weighted BH procedure with
piecewise constant weights.  Boca and Leek (2018) extended the idea by using a regression approach to estimate weights. Another related method (named AdaPT) was
proposed in Lei and Fithian (2018), which iteratively estimates the
p-value thresholds using partially censored p-values. The above
procedures can be viewed to some extent as different variants of the
weighted BH procedure. Along a separate line, Local FDR (LFDR) based
approaches have been developed to accommodate various forms of
auxiliary information. For example, Cai and Sun (2009) considered
multiple testing of grouped hypotheses using the pooled LFDR
statistic. Sun et al. (2015) developed a LFDR-based procedure to
incorporate spatial information. 
Scott et al. (2015) and Tansey et al. (2017) proposed EM-type algorithms to estimate the LFDR by
taking into account covariate and spatial information, respectively.

Although the approaches mentioned above excel in certain aspects, a method that is flexible, robust, powerful and computationally efficient is still lacking. For example, IHW developed in Ignatiadi et al. (2016) cannot handle multiple covariates. AdaPT in Lei and Fithian (2018) is computationally intensive and may suffer from significant power loss when the signal is sparse, and covariate is not very informative. Li and Barber (2017)'s procedure is not Bayes optimal as shown in Lei and Fithian (2018) and thus could lead to suboptimal power as observed in our numerical studies. The FDR regression method proposed in Scott et al. (2015) lacks a rigorous FDR control theory.  Table 1 provides a detailed comparison of these methods.

In this paper, in addition to a thorough evaluation of these methods using comprehensive simulations covering different signal structures, we propose a new procedure to incorporate covariate information with generic applicability. 
The covariates can be any continuous or categorical variables that are thought to be informative of the statistical properties of the hypothesis tests. The main contributions of our paper are two-fold:
\begin{enumerate}
\item Given a sequence of p-values $\{p_1,\dots,p_m\}$, we introduce a general decision rule of the form
\begin{equation}\label{rej0}
(1-k_i)p^{-k_i}_i\geq \frac{(1-t)\pi_{i}}{t(1-\pi_{i})},\quad
0<k_i<1, \quad 1\leq i\leq m,
\end{equation}
which serves as a surrogate for the optimal decision rule
derived under the two-component mixture model with varying mixing
probabilities and alternative densities. Here $\pi_i$ and $k_i$ are parameters that can be estimated from the covariates and p-values, and $t$ is a cut-off value to be determined by our FDR control method. We develop a new procedure
to estimate $(k_i,\pi_i)$ and find the optimal threshold value for $t$ in (\ref{rej0}).
We show that (i) when $\pi_i$ and $k_i$ are chosen
independently of the p-values, the proposed procedure provides
finite sample FDR control; (ii) our procedure provides
asymptotic FDR control when $\pi_i$ and $k_i$ are chosen to maximize a potentially misspecified likelihood based on the covariates and p-values; (iii) Similar to some recent works (e.g., Ignatiadi et al., 2016; Lei and Fithian, 2017; Li and Barber, 2017), our method allows
the underlying likelihood ratio model to be misspecified. { A distinctive feature is that our asymptotic analysis does not require the p-values to be marginally independent or conditionally independent given the covariates. More specifically, we allow the pairs of p-value and covariate across different hypotheses to be strongly mixing as specified in Assumption \ref{ass-ad3}.}

\item We develop an efficient algorithm to estimate $\pi_i$ and $k_i$.
The developed algorithm is scalable to problems with millions of tests.
Through extensive numerical studies, we show that our procedure is highly competitive to
several existing approaches in the recent literature in terms of finite sample performance.
The proposed procedure is implemented in the R package \texttt{CAMT}.
\end{enumerate}

Our method is related to Lei and Fithian (2018), and it is worth highlighting the differences from their work.
(i) Lei and Fithian (2018) uses partially censored p-values to determine the threshold, which can discard useful information concerning the alternative distribution of p-values (i.e., $f_{1,i}$ in (\ref{m1}) below) since small p-values
that are likely to be generated from the alternative are censored.
In contrast, we use all the p-values to determine the threshold. Our method is seen to exhibit more power as compared to Lei and Fithian (2018) when signal is (moderately) sparse.
Although our method no longer offers theoretical finite sample FDR control, we show empirically that the power gain is not at the cost of FDR control.
(ii) Different from Lei and Fithian (2018) which requires multiple stages for practitioners to make their final decision, our method is a single-stage procedure that only needs to be run one time; Thus the implementation of our method is faster and scalable to modern big datasets.
(iii) Our theoretical analysis is entirely different from those in Lei and Fithian (2018). In particular, we show that our method achieves asymptotic FDR control even when the p-values are dependent.

\section{Methodology}\label{sec2}
\subsection{Rejection rule}
We consider simultaneous testing of $m$ hypotheses $H_i$ for
$i=1,2,\dots,m$. Let $p_i$ be the p-value associated with the $i$th
hypothesis, and with some abuse of notation, let $H_i$ indicate the underlying truth of the $i$th
hypothesis. In other words, $H_i=0$ if the $i$th hypothesis is true
and $H_i=1$ otherwise. For each hypothesis, we observe a
covariate $x_i$ lying in some space $\mathcal{X}\subseteq
\mathbb{R}^{q}$ with $q\geq 1$. From a Bayesian viewpoint, we can model $H_i$ given $x_i$ as a
Bernoulli random variable with success probability $1-\pi_{0i}$,
where $\pi_{0i}$ denotes the prior probability that the $i$th
hypothesis is under the null when conditioning on $x_i$. One approach to model the p-value distribution is via a two-component mixture model,
\begin{align}
&H_i|x_i \sim ~\text{Bernoulli}(1-\pi_{0i}),  \label{m0}\\
&p_i|x_i,H_i \sim ~(1-H_i)f_0+H_if_{1,i}, \label{m1}
\end{align}
where $f_0$ and $f_{1,i}$ are the density functions corresponding to
the null and alternative hypotheses respectively. { In the following discussions, we shall assume that $f_0$ satisfies the following condition: for any $a\in [0,1]$
\begin{align}\label{eq-con-f0}
\int^{a}_{0}f_0(x)dx \leq \int^{1}_{1-a}f_0(x)dx.
\end{align}
This condition relaxes the assumption of uniform distribution on the unit interval. It is fulfilled when $f_0$ is non-decreasing or $f_0$ is symmetric about 0.5 (in which case the equality holds in (\ref{eq-con-f0})). We demonstrate that this relaxation is capable of describing plausible data generating processes that would create a non-uniform null distribution. Let $T$ be a test statistic such that under the null its z-score $Z=(T-\mu_0)/\sigma_0$ is standard normal. In practice, one uses $\hat{\mu}$ and $\hat{\sigma}$ to estimate $\mu_0$ and $\sigma_0$ respectively. Let $\Phi$ be the standard normal CDF. The corresponding one-sided p-value is  given by $\Phi((T-\hat{\mu})/\hat{\sigma})$ whose distribution function is $P(\Phi((T-\hat{\mu})/\hat{\sigma})\leq x)=\Phi((\Phi^{-1}(x)\hat{\sigma}+\hat{\mu}-\mu_0)/\sigma_0)$. When $\mu_0\geq \hat{\mu}$ (i.e., we underestimate the mean), one can verify that $f_0$ is a non-decreasing. In the case of $\mu_0=\hat{\mu}$ and $\sigma_0 \neq \hat{\sigma}$, $f_0$ is non-uniformly symmetric about 0.5.
}

Compared to the
classical two-component mixture model, the varying null probability
reflects the relative importance of each hypothesis given the
external covariate information $x_i$ and the varying alternative density $f_{1,i}$
emphasizes the heterogeneity among signals.  { In the context without covariate information, it is well known that the optimal rejection is based on the LFDR, see e.g., Efron (2004) and Sun and Cai (2007). The result has been generalized to the setups with group or covariate information, see e.g., Cai and Sun (2009) and Lei and Fithian (2018).} Based on these insights, one can indeed show that the optimal rejection rule that
controls the expected number of false positives while maximizes the
expected number of true positives takes the form of
\begin{equation}\label{rule}
\frac{f_{1,i}(p_i)}{f_0(p_i)}\geq \frac{(1-t)\pi_{0i}}{t(1-\pi_{0i})},
\end{equation}
where $t\in (0,1)$ is a cut-off value. This decision rule is generally unobtainable because
$f_{1,i}$ is unidentifiable without extra assumptions on
its form. Moreover, consistent estimation of the decision rule (\ref{rule}) is difficult, and even with the use of additional approximations, such as splines or piecewise constant functions.
In this work, we do not aim to estimate the optimal rejection rule directly. Instead, we try to find
a rejection rule that can mimic some useful operational characteristics of the optimal rule.
Our idea is to first replace $f_{1,i}/f_0$ by a
surrogate function $h_i$. We emphasize that $h_i$ needs not agree with
the likelihood ratio $f_{1,i}/f_0$ for our method to be valid. In fact, the validity of our method does not rely on the correct specification of model (\ref{m0})-(\ref{m1}).
We require $h_i$ to satisfy (i) $h_i(p)\geq 0$ for $p\in[0,1]$; (ii)
$\int^{1}_{0}h_i(p)dp=1$; (iii) $h$ is decreasing. Requirement (iii) is
imposed to mimic the common likelihood ratio assumption in the
literature, see e.g. Sun and Cai (2007). In this paper, we suggest to use the beta density,
\begin{align}\label{eq-LR}
h_i(p)=(1-k_i)p^{-k_i},\quad 0<k_i<1,
\end{align}
where $k_i$ is a parameter that depends on $x_i$. { Suppose that under the null hypothesis, $p_i$ is uniformly distributed, whereas under the alternative, it follows a beta distribution with parameters $(1-k_i,1)$, then
the true likelihood ratio would take exactly the form given in (\ref{eq-LR}). To demonstrate the approximation of the proposed surrogate likelihood ratio to the actual likelihood ratio for realistic problems, we simulated two binary variables and generated four alternative distributions $f_{1, i}$ depending on the four levels of the two variables (details in the legend of Figure  \ref{fig:sim:0}).  We used the proposed procedure to find the best $k_i$ and compared the CDF of the empirical distribution (reflecting the actual likelihood ratio) to that of the fitted beta distribution (reflecting the surrogate likelihood ratio). We can see from  Figure  \ref{fig:sim:0} the approximation was reasonably well and the accuracy increases with the signal density. }

Based on the surrogate likelihood ratio, the corresponding rejection rule is given by
\begin{align}\label{rej1}
h_i(p_i)\geq w_i(t):=\frac{(1-t)\pi_{i}}{t(1-\pi_{i})},
\end{align}
for some weights $\pi_i$ to be determined later. See Section \ref{alg} for more details about the estimation of $k_i$ and $\pi_i.$

\subsection{Adaptive procedure}
We first note that the false discovery proportion (FDP) associated
with the rejection rule (\ref{rej1}) is equal to
\begin{align*}
\text{FDP}(t):=&\frac{\sum^{m}_{i=1}(1-H_i)\mathbf{1}\{h_i(p_i)\geq
w_i(t)\}}{1\vee\sum^{m}_{i=1}\mathbf{1}\{h_i(p_i)\geq w_i(t)\}}.
\end{align*}
{Then for a cut-off value $t$, we have
\begin{align*}
\text{FDP}(t)=&\frac{\sum^{m}_{i=1}(1-H_i)\mathbf{1}\{p_i\leq h_i^{-1}(w_i(t))\}}{1\vee \sum^{m}_{i=1}\mathbf{1}\{h_i(p_i)\geq w_i(t)\}}
\\ \approx & \frac{\sum^{m}_{i=1}(1-H_i)P(p_i\leq h_i^{-1}(w_i(t)))}{1\vee \sum^{m}_{i=1}\mathbf{1}\{h_i(p_i)\geq w_i(t)\}}
\\ \leq & \frac{\sum^{m}_{i=1}(1-H_i)P(1-p_i\leq h_i^{-1}(w_i(t)))}{1\vee \sum^{m}_{i=1}\mathbf{1}\{h_i(p_i)\geq w_i(t)\}}
\\ \approx & \frac{1+\sum^{m}_{i=1}(1-H_i)\mathbf{1}\{h_i(1-p_i)\geq w_i(t)\}}{1\vee \sum^{m}_{i=1}\mathbf{1}\{h_i(p_i)\geq w_i(t)\}}
\\ \leq & \frac{1+\sum^{m}_{i=1}\mathbf{1}\{h_i(1-p_i)\geq w_i(t)\}}{1\vee \sum^{m}_{i=1}\mathbf{1}\{h_i(p_i)\geq w_i(t)\}}:=\text{FDP}_{\text{up}}(t),
\end{align*}
where the approximations are due to the law of large numbers and the inequality follows from Condition (\ref{eq-con-f0}).}\footnote{Rigorous theoretical justifications are provided in Theorem \ref{thm-add} and Theorem \ref{thm}.}
This strategy is partly motivated by the recent distribution-free method proposed in Barber and Cand\`{e}s (2015). We refer any FDR estimator constructed using this strategy as the BC-type estimator.
Both the adaptive procedure in Lei and Fithian (2018) and the proposed method fall into this category.
A natural idea is
to select the largest threshold such that
$\text{FDP}_{\text{up}}(t)$ is less or equal to a prespecified FDR level $\alpha.$
Specifically, we define
\begin{align*}
t^*=\max\left\{t\in [0,t_{\text{up}}]: \text{FDP}_{\text{up}}(t)=\frac{1+\sum^{m}_{i=1}\mathbf{1}\{h_i(1-p_i)\geq
w_i(t)\}}{1\vee\sum^{m}_{i=1}\mathbf{1}\{h_i(p_i)\geq
w_i(t)\}}\leq
\alpha\right\},
\end{align*}
where $t_{\text{up}}$ satisfies that $w_i(t_{\text{up}})\geq h_i(0.5)$ for all $i,$
and we reject all hypotheses such that $h_i(p_i)\geq w_i(t^*)$. The
following theorem establishes the finite sample control of the above
procedure when $\pi_i$ and $h_i$ are prespecified and thus independent of the p-values. For example, $\pi_i$ and $h_i$ are estimated based on data from an independent but related study.
\begin{theorem}\label{thm-add}
Suppose $h_i$ is strictly decreasing for each $i$ and $f_0$ satisfies Condition (\ref{eq-con-f0}).
If the p-values are independent and the choice of $h_i$ and $\pi_i$ is independent of the p-values, then the adaptive procedure provides finite sample FDR control at level
$\alpha$.
\end{theorem}

\subsection{An algorithm}\label{alg}
The optimal choices of $\pi_i$ and $k_i$ are rarely known in practice, and a generally applicable data-driven method is desirable.
In this subsection, we propose an EM-type algorithm to estimate $\pi_i$ and $k_i$. In particular, we model
both $\pi_i$ and $k_i$ as functions of the covariate $x_i$. As an illustration, we
provide the following example.
\begin{example}\label{example}
{\rm Suppose
\begin{align*}
&p_{i}|x_i,H_i \sim (1-H_i)f_0+H_if_{1,i},
\\&x_{i}|H_i \sim (1-H_i)g_0+H_ig_1,
\end{align*}
where $H_i\sim^{\text{i.i.d}}\text{Bernoulli}(1-\pi_{0})$.
Using the Bayes rule, we have
\begin{align*}
f(p_{i}|x_{i})=&\frac{f(p_{i},x_{i}|H_i=0)\pi_0+f(p_{i},x_{i}|H_i=1)(1-\pi_0)}{f(x_{i}|H_i=0)\pi_{0}+f(x_{i}|H_i=1)(1-\pi_0)}
\\=&\frac{f(p_{i}|x_i,H_i=0)f(x_{i}|H_i=0)\pi_0+f(p_{i}|x_i,H_i=1)f(x_{i}|H_i=1)(1-\pi_0)}{f(x_{i}|H_i=0)\pi_{0}+f(x_{i}|H_i=1)(1-\pi_0)}
\\=&\pi(x_{i})f_{0}(p_{i})+(1-\pi(x_{i}))f_{1,i}(p_{i}),
\end{align*}
where $\pi(x)=g_0(x)\pi_0/\{g_0(x)\pi_{0}+g_1(x)(1-\pi_0)\}=f(H_i=0|x_i=x).$
Therefore, $\pi_i$ is the conditional probability that the $i$th
hypothesis is under the null given the covariate $x_i$. }
\end{example}

{ To motivate our estimation procedure for $\pi_i$ and $k_i$, let us define $\pi_{\theta}(x)=1/(1+e^{-\theta_0-\theta_1'x})$ and $k_{\beta}(x)=1/(1+e^{-\beta_0-\beta_1'x})$ for $x\in\mathbb{R}^q$, where $\theta=(\theta_0,\theta_1)$ and $\beta=(\beta_0,\beta_1)$. Suppose that conditional on $x_i$ and marginalizing over $H_i$,
\begin{align*}
f(p_i|x_i)=&\pi_{\theta}(x_i)f_0(p_i)+(1-\pi_{\theta}(x_i))f_{1,i}(p_i)
=f_0(p_i)\left\{\pi_{\theta}(x_i)+(1-\pi_{\theta}(x_i))\frac{f_{1,i}(p_i)}{f_0(p_i)}\right\}.
\end{align*}
Replacing $f_{1,i}/f_0$ by the surrogate likelihood ratio whose parameters $k_i$ depend on $x_i$, we obtain
$$\tilde{f}(p_i|x_i)=f_0(p_i)\left\{\pi_{\theta}(x_i)+(1-\pi_{\theta}(x_i))(1-k_{\beta}(x_i))p_i^{-k_\beta(x_i)}\right\}.$$
Moving to a log scale and summing up the individual log likelihoods, we see that the null density is a nuisance parameter that does not depend on $\theta$ and $\beta$:
$$\sum^{m}_{i=1}\log \tilde{f}(p_i|x_i)=\sum^{m}_{i=1}\log\left\{\pi_{\theta}(x_i)+(1-\pi_{\theta}(x_i))(1-k_{\beta}(x_i))p_i^{-k_\beta(x_i)}\right\}+C_0,$$
where $C_0=\sum_{i=1}^m \log f_0(p_i)$.} The above discussions thus motivate the following optimization problem for estimating the unknown parameters:
\begin{align}\label{eq-mle}
\max_{\theta=(\theta_0,\theta_1)'\in \Theta,\beta=(\beta_0,\beta_1)'\in \mathcal{B}}\sum^{m}_{i=1}\log\{\pi_i+(1-\pi_i)(1-k_i)p^{-k_i}\},
\end{align}
where
\begin{align}\label{eq-par}
& \log\left(\frac{\pi_i}{1-\pi_i}\right)=\theta_0+\theta_1'x_i, \quad \log\left(\frac{k_i}{1-k_i}\right)=\beta_0+\beta_1'x_i,
\end{align}
and $\Theta,\mathcal{B}\subseteq \mathbb{R}^{q+1}$ are some compact parameter spaces. This problem can be solved using the EM-algorithm together with the Newton's method in its M-step.
Let $\hat{\theta}$ and $\hat{\beta}$ be the maximizer from (\ref{eq-mle}). Define
$$\hat{\pi}_i=W(1/(1+e^{-\tilde{x}_i'\hat{\theta}}),\epsilon_1,\epsilon_2):=\begin{cases}
                \epsilon_1, & \mbox{if } 1/(1+e^{-\tilde{x}_i'\hat{\theta}})\leq \epsilon_1, \\
                1/(1+e^{-\tilde{x}_i'\hat{\theta}}), & \mbox{if } \epsilon_1<1/(1+e^{-\tilde{x}_i'\hat{\theta}})<1-\epsilon_2, \\
                1-\epsilon_2, & \mbox{otherwise},
              \end{cases}$$
and $\hat{k}_i=1/(1+e^{-\tilde{x}_i'\hat{\beta}})$ with $\tilde{x}_i=(1,x_i')'$, and
$$\hat{w}_i(t)=\frac{(1-t)\hat{\pi}_i}{t(1-\hat{\pi}_i)}.$$
We use winsorization to prevent $\hat{\pi}_i$ from being too close to zero. In numerical studies, we found the choices of $\epsilon_1=0.1$
and $\epsilon_2=10^{-5}$ perform reasonably well. Further denote
\begin{align*}
\hat{t}=\max\left\{t\in [0,1]:
\frac{1+\sum^{m}_{i=1}\mathbf{1}\{(1-\hat{k}_i)(1-p_i)^{-\hat{k}_i}>
\hat{w}_i(t)\}}{1\vee\sum^{m}_{i=1}\mathbf{1}\{(1-\hat{k}_i)p_i^{-\hat{k}_i}\geq
\hat{w}_i(t)\}}\leq \alpha\right\}.
\end{align*}
Then we reject the $i$th hypothesis if
$$(1-\hat{k}_i)p_i^{-\hat{k}_i}\geq
\hat{w}_i(\hat{t}).$$
\begin{rem}
{\rm
We can replace $x_i\in\mathbb{R}^q$ by $(g_1(x_i),\dots,g_{q_0}(x_i))\in\mathbb{R}^{q_0}$ for some transformations $(g_1,\dots,g_{q_0})$ to allow nonlinearity in the logistic regressions. In numerical studies, we shall consider the spline transformation.
}
\end{rem}


\section{Asymptotic results}\label{sec:asy}
\subsection{FDR control}
In this subsection, we provide asymptotic justification for the proposed procedure. Note that
\begin{align*}
\mathbf{1}\{(1-\hat{k}_i)p^{-\hat{k}_i}\geq
\hat{w}_i(t)\}=\mathbf{1}\{p\leq c(t,\hat{\pi}_{i},\hat{k}_i)\}\text{
for }
c(t,\hat{\pi}_{i},\hat{k}_i)=1\wedge\left\{\frac{t(1-\hat{k}_i)(1-\hat{\pi}_{i})}{(1-t)\hat{\pi}_{i}}\right\}^{1/\hat{k}_i}.
\end{align*}
Define
\begin{align*}
\text{FDR}(t,\Pi,K)=E\left[\frac{\sum_{i=1}^{m}(1-H_i)\mathbf{1}\{p_i\leq
c(t,\pi_{i},k_i)\}}{\sum_{i=1}^{m}\mathbf{1}\{p_i\leq
c(t,\pi_{i},k_i)\}}\right]
\end{align*}
with $\Pi=(\pi_1,\dots,\pi_m)$ and $K=(k_1,\dots,k_m)$. We make the following assumptions to
facilitate our theoretical derivations.


\begin{ass}\label{ass-ad1}
Suppose the parameter spaces $\Theta$ and $\mathcal{B}$ are both compact.
\end{ass}

\begin{ass}\label{ass-ad2}
Suppose
$$\lim_m\frac{1}{m}\sum^{m}_{i=1}E\log\{\pi_{\theta}(x_i)+(1-\pi_{\theta}(x_i))(1-k_{\beta}(x_i))p_i^{-k_\beta(x_i)}\}$$
converges uniformly over $\theta\in\Theta$ and $\beta\in\mathcal{B}$ to $R(\theta,\beta)$, which has a unique
maximum at $(\theta^*,\beta^*)$ in $\Theta\times \mathcal{B}.$
\end{ass}

Let $\mathcal{F}_a^b=\sigma((x_i,p_i),a\leq i\leq b)$ be the Borel $\sigma$-algebra generated by the random variables $(x_i,p_i)$
for $a\leq i \leq b$. Define the $\alpha$-mixing and $\phi$-mixing coefficients respectively as
\begin{align*}
&\alpha(v)=\sup_b\sup_{A\in\mathcal{F}_{-\infty}^b,B\in\mathcal{F}_{b+v}^{+\infty}}|P(AB)-P(A)P(B)|,
\\&\phi(v)=\sup_b\sup_{A\in\mathcal{F}_{-\infty}^b,B\in\mathcal{F}_{b+v}^{+\infty},P(B)>0}|P(A|B)-P(A)|.
\end{align*}
\begin{ass}\label{ass-ad3}
Suppose $(x_i,p_i)$ is $\alpha$-mixing with $\alpha(v)=O(v^{-\xi})$ for $\xi>r/(r-1)$ and $r>1$ (or $\phi$-mixing with $\phi(v)=O(v^{-\xi})$ for $\xi>r/(2r-1)$ and $r\geq 1$).
Further assume $\sup_i E|\log(p_i)|^{r+\delta}<\infty$ and $\max_i\|x_i\|_{\infty}<C$, where $\|\cdot\|_{\infty}$ denotes the $l_{\infty}$ norm of a vector and $C,\delta>0$.
\end{ass}


Assumption \ref{ass-ad1} is standard. Assumption \ref{ass-ad2} is a typical condition in
the literature of maximum likelihood estimation for misspecified
models, see e.g. White (1982). Assumption \ref{ass-ad3} relaxes the usual independence assumption by allowing $(x_i,p_i)$ to be weakly dependent.
It is needed to establish the uniform strong law of large numbers for the process $R_m(\theta,\beta)$ defined in the proof of Lemma \ref{lem-30} below which establishes the uniform strong consistency for $\hat{\pi}_i$ and $\hat{k}_i$. The boundedness assumption on $x_i$ could be relaxed with a more
delicate analysis to control its tail behavior and study the convergence rate of $\hat{\theta}$ and $\hat{\beta}$.
Denote by $\|\cdot\|$ the $l_{2}$ norm of a vector. An essential condition required in our proof of Lemma \ref{lem-30} is $\|\hat{\theta}-\theta^*\|\max_{1\leq i\leq n}\|x_i\|=o_{a.s.}(1)$. { If $\|\hat{\theta}-\theta^*\|=O_{a.s.}(n^{-a})$ for some $a>0$, then by the Borel-Cantelli lemma, we require $\max_{1\leq i\leq n} E\|x_i\|^{k}<\infty$ for some $k$ with $ak>2,$ i.e., $x_i$ should have a sufficiently light polynomial tail.} We remark that Assumption \ref{ass-ad3} can be replaced by more primitive conditions which allow other weak dependence conditions,
see, e.g., P\"{o}tscher and Prucha (1989). Let $\pi_i^*=W(1/(1+e^{-\tilde{x}_i'\theta^*}),\epsilon_1,\epsilon_2)$
and $k_i^*=1/(1+e^{-\tilde{x}_i'\beta^*})$.
\begin{lemma}\label{lem-30}
Under Assumptions \ref{ass-ad1}-\ref{ass-ad3}, we have $$\max_{1\leq i\leq
m}|\hat{\pi}_i-\pi_i^*|\rightarrow^{a.s.} 0,\quad \max_{1\leq i\leq
m}|\hat{k}_i-k_i^*|\rightarrow^{a.s.} 0.$$
\end{lemma}

We impose some additional assumptions to study the asymptotic FDR control.

\begin{ass}\label{ass-31}
For two sequences $a_i,b_i \in [\epsilon,1]$ with small enough $\epsilon$ and large enough $m$,
\begin{align*}
&\left|\frac{1}{m}\sum_{i=1}^{m}\left\{P(p_i\leq a_i|x_i)-P(p_i\leq
b_i|x_i)\right\}\right| \leq c_0\max_{1\leq i\leq m}|a_i-b_i|,
\end{align*}
where $c_0$ depends on $\epsilon$ but is independent of $m, x_i, a_i$ and
$b_i$.
\end{ass}

\begin{ass}\label{ass-32}
Assume that
\begin{align}
&\frac{1}{m}\sum_{i=1}^{m}P(p_i\leq
c(t,\pi_{i}^*,k^*_i))\rightarrow G_0(t),\\
&\frac{1}{m}\sum_{i=1}^{m}P(1-p_i<
c(t,\pi_{i}^*,k^*_i))\rightarrow G_1(t),\\
&\frac{1}{m}\sum_{H_i=0}P(p_i\leq
c(t,\pi_{i}^*,k^*_i))\rightarrow \tilde{G}_1(t),
\end{align}
for any $t\geq t_0$ with $t_0>0$, where $G_0(t)$, $G_1(t)$ and
$\tilde{G}_1(t)$ are all continuous functions of $t$. Note that the probability here is taken with respect to the joint distribution of $(p_i,x_i)$.
\end{ass}

Let $U(t)=G_1(t)/G_0(t)$, where $G_1$ and $G_0$ are defined in
Assumption \ref{ass-32}.

\begin{ass}\label{ass-33}
There exists a $t'>t_0>0$ such that $U(t')<\alpha.$
\end{ass}

Assumption \ref{ass-31} is fulfilled if the conditional density of $p_i$ given $x_i$ is
bounded uniformly across $i$ on $[\epsilon,1]$. This assumption
is not very strong as we still allow the density to be unbounded
around zero. Assumptions \ref{ass-32}-\ref{ass-33} are similar to those in Theorem 4 of Storey et al. (2004).
In particular, Assumption \ref{ass-33} ensures the existence of a cut-off to control the FDR
at level $\alpha.$

We are now in position to state the main result of this section
which shows that the proposed procedure provides asymptotic FDR
control. The proof is deferred to the supplementary material.
\begin{theorem}\label{thm}
Suppose Assumptions \ref{ass-ad1}-\ref{ass-33} hold and $f_0$ satisfies Condition (\ref{eq-con-f0}). Then we have
$$\limsup_{m}\text{FDR}(\hat{t},\hat{\Pi},\hat{K})\leq \alpha,$$
where $\hat{\Pi}=(\hat{\pi}_1,\dots,\hat{\pi}_m)$ and $\hat{K}=(\hat{k}_1,\dots,\hat{k}_m)$.
\end{theorem}
It is worth mentioning that the validity of our method does not rely on the mixture model assumption (\ref{m0})-(\ref{m1}). In this sense, our method is misspecification robust as the classical BH procedure does.
We provide a comparison between our method and some recently proposed approaches in the following table.

\begin{table}[!ht]\label{tab}
\small
\centering
\begin{tabular}{p{2.5cm}p{1cm}p{1.5cm}p{2.2cm}p{1.75cm}p{1.6cm}p{1.6cm}p{1.6cm}}
\hline
Procedure &  $\pi_0$   & $f_1$ & FDR control & Dependent p-values & Misspec. robust & Multiple covariates &  Computation\\
\hline
Ignatiadis et al. (2016) & Varying & Partially used & Asymptotic control & Unknown &  Yes & No &  ++++\\
\hline
Li and Barber (2017) & Varying & Not used & Finite sample upper bound & Gaussian copula & Yes &  No$^*$ & ++++\\
\hline
Lei and Fithian (2016) & Varying & Varying & Finite sample control & Unknown & Yes  & Yes & +\\
\hline
Scott et al. (2015) & Varying   & Fixed & No guarantee &  Unknown & Unknown  & Yes & +++\\
\hline
Boca and Leek (2018) & Varying   & Not used & Unknown &  Unknown & Yes & Yes & +++\\
\hline
\textbf{The proposed method} & Varying   & Varying & Asymptotic control &  Asymptotic &  Yes &Yes & +++\\
\hline
\end{tabular}
\\ \caption{Comparison of several covariate adaptive FDR control procedures in recent literature. The number of ``+'' represents the speed. *The framework of Li and Barber (2017) allows accommodating multiple covariates, but the provided software did not implement.}
\end{table}

\subsection{Power analysis}
We study the asymptotic power of the oracle procedure. Suppose the mixture model (\ref{m0})-(\ref{m1}) holds with
$\pi_{0i}=\pi_0(x_i)$ and $f_{1,i}(\cdot)=f_{1}(\cdot;x_i)$, where $f_{1}(\cdot;x)$ is a density function for any fixed $x\in\mathcal{X}$.
Denote by $F_1(\cdot;x)$ and $\bar{F}_1(\cdot;x)$ the distribution and survival functions associated with $f_1(\cdot;x)$ respectively.
Suppose the empirical distribution of $x_i$'s converges to the probability law $\mathcal{P}$. Consider the oracle procedure with $\pi_i=\pi_0(x_i)$ and $k_i=k_0(x_i)$. Here $k_0(\cdot)$ minimizes the integrated Kullback-Leibler divergence, i.e.,
\begin{align*}
&k_0=\argmin_{k\in\mathcal{K}}\int\text{D}_{\text{KL}}(f(\cdot;x)||g(;k(x)))\mathcal{P}(dx),
\\&\text{D}_{\text{KL}}(f(\cdot;x)||g(\cdot;k(x)))=\int^{1}_{0} f(p;x)\log\frac{f(p;x)}{g(p;k(x))}dp,
\end{align*}
with $f(p;x)=\pi_0(x)f_0(p)+(1-\pi_0(x))f_{1}(p;x)$ and $g(p;k(x))=\pi_0(x)+(1-\pi_0(x))(1-k(x))p^{-k(x)}$, and $\mathcal{K}=\left\{k(x):\log\left(\frac{k(x)}{1-k(x)}\right)=\beta_0+\beta_1'x,(\beta_0,\beta_1)\in\mathcal{B}\right\}$.
Write $c(t,x)=c(t,\pi_0(x),k_0(x))$. By the law of large numbers, the realized power of the oracle procedure has the
approximation
\begin{align*}
\text{Power}=&\frac{\sum_{i=1}^{m}\mathbf{1}\{i:H_i=1,p\leq c(t,x_i)\}}{\sum_{i=1}^{m}\mathbf{1}\{i:H_i=1\}}
\approx
\frac{\int(1-\pi(x))F_1(c(t_{\text{opt}},x);x)\mathcal{P}(dx)}{\int(1-\pi(x))\mathcal{P}(dx)},
\end{align*}
where $t_{\text{opt}}$ is the largest positive number such that
\begin{align}\label{eq-pow}
\frac{\int \{\pi_0(x)F_0(c(t,x))+(1-\pi_0(x))\bar{F}_1(1-c(t,x);x) \}\mathcal{P}(dx)}{\int \{\pi_0(x)F_0(c(t,x))+(1-\pi_0(x))F_1(c(t,x);x) \}\mathcal{P}(dx)}\leq
\alpha.
\end{align}
We remark that when
\begin{align}\label{eq-pow1}
\frac{\int (1-\pi_0(x))\bar{F}_1(1-c(t_{\text{opt}},x);x)\mathcal{P}(dx)}{\int \{\pi_0(x)F_0(c(t_{\text{opt}},x))+(1-\pi_0(x))F_1(c(t_{\text{opt}},x);x) \}\mathcal{P}(dx)}\approx 0,
\end{align}
the asymptotic power of the proposed procedure is closed to the oracle procedure based on the LFDR given by
\begin{align}\label{lfdr}
\text{LFDR}_i(p_i)=\frac{\pi_{0i}f_0(p_i)}{\pi_{0i}f_0(p_i)+(1-\pi_{0i})f_{1,i}(p_i)}.
\end{align}

\section{Simulation  studies}
We conduct comprehensive simulations to evaluate the finite-sample performance of the proposed method  and compare it to competing methods. For genome-scale multiple testing, the numbers of hypotheses could range from thousands to millions. For demonstration purpose, we start with $m{=}10,000$ hypotheses. To study the impact of signal density and strength,  we simulate three levels of signal density (sparse, medium and dense signals) and six levels of signal strength (from very weak to very strong).  To  demonstrate the power improvement by using external covariates, we simulate covariates of varying informativeness (non-informative, moderately informative and strongly informative).  For simplicity, we simulate one covariate $x_i  \sim N(0, 1)$ for $i=1,\cdots,m$. Given $x_i$, we let
$$\pi_{0i} = \frac{\exp (\eta_i)}{1 + \exp(\eta_i)},  ~~~ \eta_i = \eta_0 + k_d x_i, $$
where $\eta_0$ and $k_d$ determine the baseline signal density and the informativeness of the covariate, respectively.  { For each simulated dataset, we fix the value of  $\eta_0$ and $k_d$. }  We set $\eta_0 \in \{3.5, 2.5, 1.5\}$, which achieves a signal density  around $3\%$, $8\%$, and $18\%$ respectively at the baseline (i.e., no covariate effect), representing sparse, medium and dense signals. We set $k_d \in \{0, 1,1.5\}$, representing a non-informative, moderately informative and strongly informative covariate, respectively.  { Thus, we have a total of $3 \times 3 = 9$ parameter settings.}
Based on $\pi_{0i}$, the underlying truth $H_i$ is simulated from
$$H_i \sim \text{Bernoulli}(1 - \pi_{0i}).$$
Finally, we simulate independent z-scores using $$z_i \sim  N(k_sH_i, 1), $$ where $k_s$ controls the signal strength (effect size) and we use values equally spaced on $[2, 2.8]$.  Z-scores are converted into p-values using the one-sided formula $1 - \Phi(z_i)$. P-values together with $x_{i}$ are used as the input for the proposed method.

In addition to the basic setting (denoted as Setup S0), we investigate other settings to study the robustness of the proposed method. Specifically, we study

 \begin{itemize}
\item[Setup S1.] {\it Additional $f_1$ distribution}. Instead of simulating normal z-scores under $f_1$, we simulate z-scores from a non-central gamma distribution with the shape parameter $k{=2}$. The scale/non-centrality parameters of the non-central gamma distribution are chosen to match the variance and mean of the normal distribution under S0.

 \item[Setup S2.] {\it Covariate-dependent $\pi_{0i}$ and $f_{1,i}$}. On top of the basic setup S0, we simulate another covariate $x_i' \sim N(0, 1)$  and let $x_i'$ affect $f_{1,i}$. Specifically, we scale $k_s$ by
 $\frac{\displaystyle 2\exp (k_f x_i')}{\displaystyle  1 + \exp(k_f x_i')}, $
 where we set $k_f \in \{0, 0.25, 0.5\}$ for non-informative, moderately informative and strongly informative covariate scenarios, respectively.

 \item[Setup S3.] {\it Dependent hypotheses}. We further investigate the effect of dependency among hypotheses by simulating correlated multivariate normal z-scores. Four correlation structures, including two block correlation structures and two AR(1) correlation structures, are investigated. For the block correlation structure, we divide the 10,000 hypotheses into 500 equal-sized blocks.  Within each block, we simulate equal positive correlations ($\rho{=}0.5$) (S3.1).  We also further divide the block into 2 by 2 sub-blocks, and simulate negative correlations ($\rho{=}-0.5$) between the two sub-blocks (S3.2).  For AR(1) structure,  we investigate both $\rho{=}0.75^{|i - j|}$ (S3.3) and $\rho{=}(-0.75)^{|i - j|}$ (S3.4).

 {
  \item[Setup S4.]  {\it Heavy-tail covariate}.  In this variant, we generate $x_i$ from the t distribution with 5 degrees of freedom.

 \item[Setup S5.] {\it Non-theoretical null distribution}.  We simulate both increasing and decreasing $f_0$.  For an increasing $f_0$ (S5.1), we generate null z-score  $z_i | H_0 \sim N(-0.15, 1)$. For a decreasing $f_0$ (S5.2), we generate null z-score  $z_i | H_0 \sim N(0.15, 1)$.

 }

  \end{itemize}
We present the simulation results for the Setup S0-S2 in the main text and the results for the Setup S3-S5 in the supplementary material. { To allow users to conveniently implement our method and reproduce the numerical results reported here, we make our code and data publicly available at https://github.com/jchen1981/CAMT}.

 \subsection{Competing methods}
We label our method as CAMT (Covariate Adaptive Multiple Testing) and compare it to the following competing methods:
\begin{itemize}
\item Oracle: Oracle procedure based on LFDR (see e.g., (\ref{lfdr})) with simulated $\pi_{0i}$ and $f_{1, i}$, which theoretically has the optimal performance;
\item BH: Benjamini-Hochberg procedure  (Benjamini et al., 1995, \textit{p.adjust} in  R 3.4.2);
\item ST: Storey's BH procedure  (Storey 2002, \textit{qvalue} package,  v2.10.0);
\item { BL: Boca and Leek procedure  (Boca and Leek, 2018, \textit{swfdr} package,  v1.4.0);}
\item IHW: Independent hypothesis weighting  (Ignatiadis et al., 2016, \textit{IHW} package, v1.6.0);
\item { FDRreg:  False discovery rate regression (Scott et al., 2015, \textit{FDRreg} package, v0.2, \textit{https://github.com/jgscott/FDRreg}), FDRreg(T) and FDRreg(E) represent FDRreg with the theoretical null and empirical null respectively; }
\item SABHA: Structure adaptive BH procedure (Li and Barber, 2017, $\tau = 0.5, \epsilon = 0.1$ and stepwise constant weights, \textit{https://www.stat.uchicago.edu/$\sim$rina/sabha/All\_q\_est\_functions.R});
\item AdaPT: Adaptive p-value thresholding procedure (Lei and Fithian, 2018, \textit{adaptMT} package, v1.0.0).
\end{itemize}

We evaluate the performance based on FDR control (false discovery proportion)  and power (true positive rate)  with a target FDR level of 5\%.  Results are averaged over 100 simulation runs.

 \subsection{Simulation results}

 We first study the performance of the proposed method under the basic setup (S0, Figure \ref{fig:sim:1}).  All compared methods generally controlled the FDR around/under the nominal level of 0.05 and no serious FDR inflation was observed at any of the parameter setting (Figure \ref{fig:sim:1}A).  { However, FDRreg exhibited a slight FDR inflation under some parameter settings and the inflation seemed to increase with the informativeness of the covariate and signal density.}   Conservativeness was also observed for some methods in some cases.   As expected, the BH procedure, which did not take into account $\pi_0$,  was conservative when the signal was dense. IHW procedure  was generally more conservative than BH and the conservativeness increased with the informativeness of the covariate.   CAMT, the proposed method, was conservative when the signal was sparse and the covariate was less informative. The conservativeness was more  evident when the effect size was small but decreased as the effect size became larger. AdaPT was more conservative than CAMT under sparse signal/weak covariate.  In terms of power (Figure \ref{fig:sim:1}B), there were several interesting observations.  First, as the covariate became more informative, all the covariate adaptive methods became more powerful than ST and BH. The power differences between these methods also increased.  Second, { FDRreg was the most powerful across settings.  Under a highly informative covariate,  it was even slightly above the oracle procedure, which theoretically had an optimal power.  The superior power of FDRreg could be partly explained by a less well controlled FDR.} The  IHW was more powerful than BL/SABHA when the signal was sparse; but the trend reversed when the signal was dense. Third, AdaPT was very powerful when the signal was dense and the covariate was highly informative. However, the power decreased as the signal became more sparse and the covariate became less informative.  In fact, when the signal was sparse and the covariate was not informative or moderately informative, AdaPT had the lowest power.  In contrast, the proposed method CAMT  was close to the oracle procedure. It  was comparable to AdaPT when AdaPT was the most powerful, but was significantly more powerful than AdaPT in its unfavorable scenarios. CAMT had a clear edge when the covariate was informative and signal was sparse. Similar to AdaPT, CAMT had some power loss under sparse signal and non-informative covariate, probably due to the discretization effect from the BC-type estimator.

 We conducted more evaluations on type I error control under S0.  We investigated the FDR control across different target levels. Figure \ref{fig:sim:2} showed excellent FDR control across target levels for all methods except FDRreg.  The actual FDR level of BH and IHW was usually below the target level. CAMT was slightly conservative at a small target level under the scenario of sparse signal and less informative covariate, but it became less conservative at larger target levels. We also simulated a complete null, where no signal was included (Figure \ref{fig:sim:3}). In such case, FDR was reduced to FWER. { Interestingly, FDRreg was as conservative as CAMT and AdaPT under the complete null.}

{ It is interesting to study the performance of the competing methods  under a much larger feature size, less signal density, and weaker signal strength, representing the most challenging scenario in real problems.  To achieve this end, we simulated $m=100,000$ features with a signal density of $0.5\%$ at the baseline (no covariate effect).  Under a moderately informative covariate, we observed a substantial power improvement of CAMT over all other methods including FDRreg while controlling the FDR adequately at different target levels (Figure \ref{fig:sim:3:1}). } We further reduced the feature size to 1,000  (Figure \ref{fig:sim:6:1} in the supplement) to study the robustness of the methods to a much smaller feature size.   Although CAMT and AdaPT were still more powerful than the competing methods  when the signal was dense and the covariate was informative, a significant power loss was observed in other parameter settings, particularly under sparse signal and a less informative covariate.  As we further decreased the feature size to 200, CAMT and AdaPT became universally less powerful than ST across parameter settings (data not shown).  Therefore, application of CAMT or AdaPT to datasets with small numbers of features was not recommended unless the signal was dense and the covariate was highly informative.

{ We also simulated datasets, where the z-scores under the alternative were drawn from a non-central gamma distribution (Setup S1).  Under such setting, the trend remained almost the same as the basic setup (Figure \ref{fig:sim:4}), but FDRreg had a more marked FDR inflation.  When both $\pi_{0i}$ and $f_{1, i}$ depended on the covariate (Setup S2),  CAMT became slightly more powerful without affecting the FDR control, especially when the covariate was highly informative  (Figure \ref{fig:sim:5}). Meanwhile, the performance of FDRreg was also remarkable with a very small FDR inflation.   However, if we increased the effect on $f_{1, i}$ by reducing the standard deviation of the z-score under the alternative, FDRreg was no longer robust and the observed FDP was  substantially above the target level when the signal strength was weak, indicating the benefit of modeling covariate-dependent $f_1$ (Figure \ref{fig:sim:6} in the supplement). }  CAMT was also robust to different correlation structures  (Setup S3.1, S3.2, S3.3, S3.4) and we observed similar performance under these correlation structures (Figures \ref{fig:sim:7}-\ref{fig:sim:10} in the supplement).   { The performance of CAMT was also robust to a heavy-tail covariate (Setup S4, Figure \ref{fig:sim:11} in the supplement). } In an unreported numerical study, we
added different levels of perturbation to the covariate by multiplying random small values drawn from Unif(0.95, 1.05), Unif(0.9, 1.1), and Unif(0.8, 1.2), respectively.
We observed that the $\pi_0$ estimates under perturbation are highly correlated with the $\pi_0$ estimates without perturbation, which showed the stability of our method against
data perturbations.

{  We also examined the robustness of CAMT to the deviation from the theoretical null (Setup S5). Specifically, we simulated both decreasing and increasing $f_0$.  The new results were presented in Figures \ref{fig:sim:12} and  \ref{fig:sim:13} in the supplement.  We observed that, for an increasing $f_0$,  all the methods other than FDRreg were conservative and had substantial less power than the oracle procedure.  FDRreg using a theoretical null was conservative when the covariate was less informative but was anti-conservative under a highly informative covariate. On the other hand, FDRreg using an empirical null had an improved power and controlled the FDR closer to the target level for most settings.  However, it did not control the FDR well when the signal was dense and the prior information was strong.  When $f_0$ was decreasing, all the methods without using the empirical null failed to control the FDR.   FDRreg with an empirical null improved the FDR control substantially for most settings but still could not control the FDR well under the dense-signal and strong-prior setting.  Therefore, there is still room for improvement to address the empirical null problem.}

 Finally, we compared the computational efficiency of these competing methods (Figure \ref{fig:sim:14}).  SABHA (step function) and IHW were computationally the most efficient and they completed the analysis for one million p-values in less than two  minutes. CAMT and the new version of FDRreg (v0.2) were also computationally efficient, followed by BL,  and they all could complete the computation in minutes for  one million p-values under S0.  AdaPT was computationally the most intensive  and completed the analysis in hours for one million p-values. { We note that all the methods including AdaPT are computationally feasible for  a typical omics dataset. }

 In summary, CAMT  improves over existing covariate adaptive multiple testing procedures, and is a powerful, robust and computationally efficient tool for large-scale multiple testing.

 \begin{figure}
\centering
\includegraphics[width=0.9\textwidth]{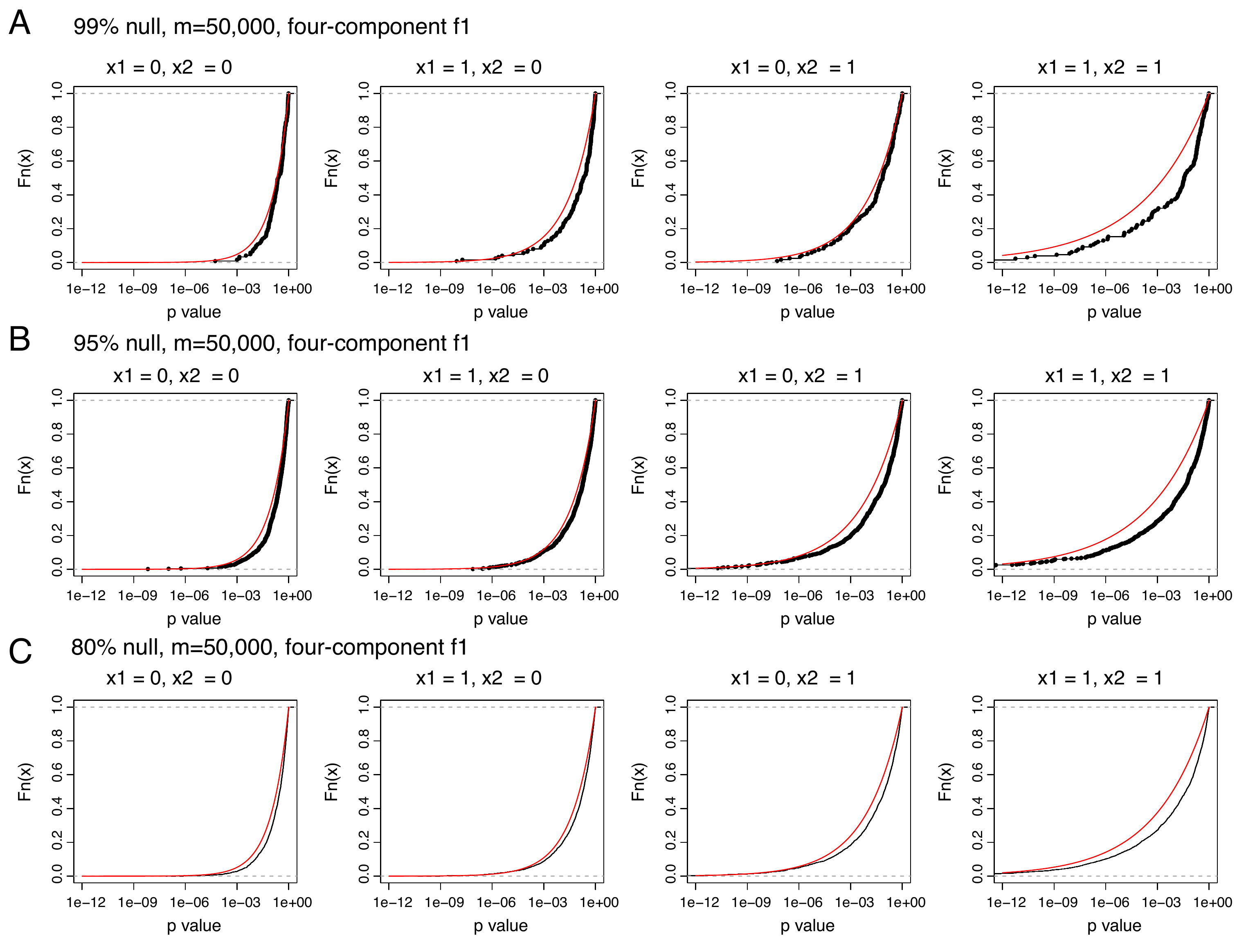}
\caption[Illustration of the fit of beta]{ Illustration of the approximation to the true likelihood ratio by the surrogate likelihood ratio based on a beta distribution.  Two binary covariates $x_1$ and $x_2$  were simulated.  The z-score under the alternative was drawn from $N(0, 1.5+0.5x_1+x_2)$.  Three levels of null proportions (A - 99\%, B - 95\%, and C - 80\%) were simulated, where the null z-score was drawn from $N(0, 1)$.  Two-sided p-values were calculated based on the z-score and the parameter $k_i$ of the beta distribution was estimated by CAMT. The CDF of the empirical distribution of the p-value under the alternative (black) was compared to CDF of the fitted beta-distribution (red). }
\label{fig:sim:0}
\end{figure}

 \begin{figure}
\centering
\includegraphics[width=0.9\textwidth]{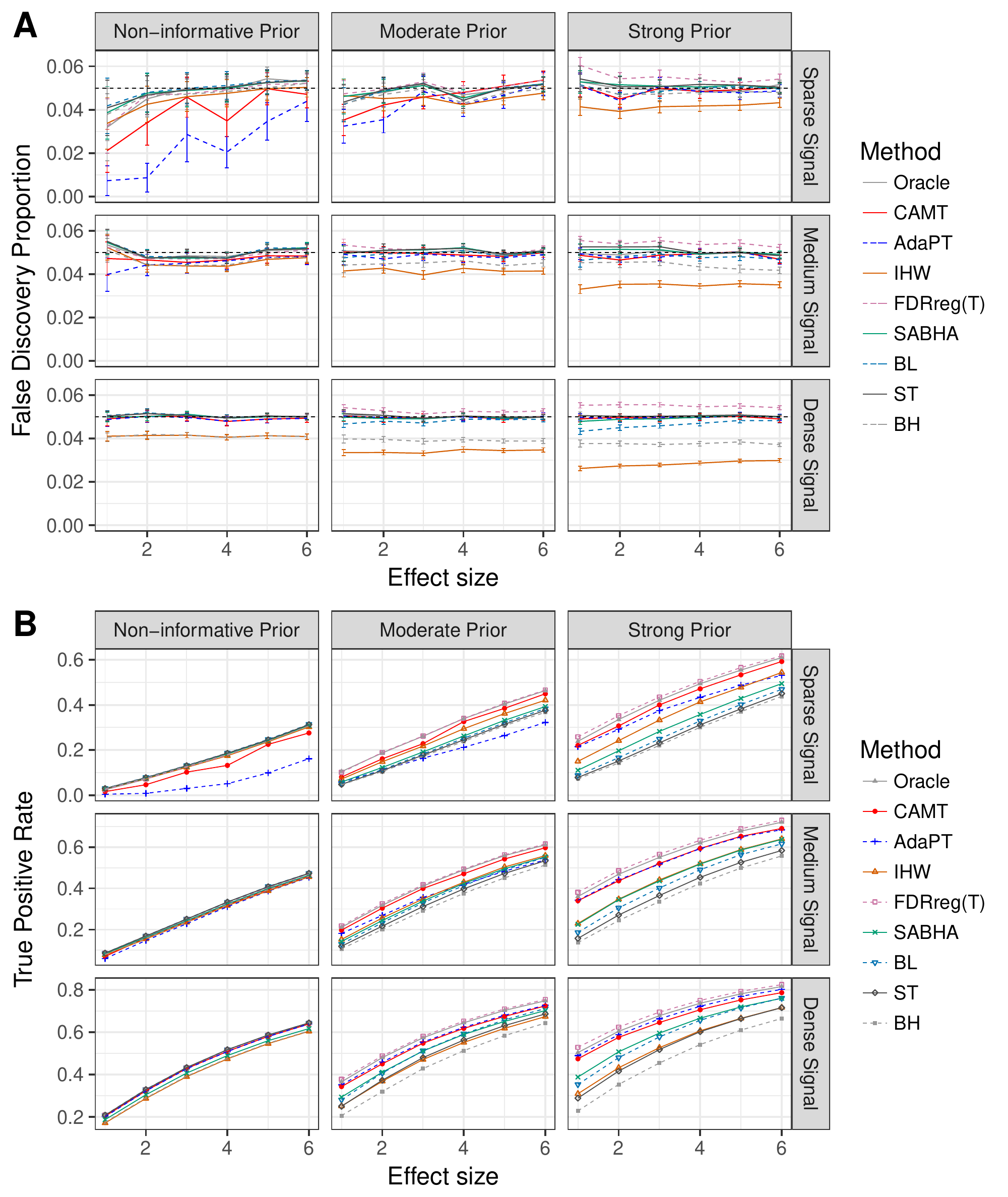}
\caption[Performance comparison under S0]{Performance comparison under the basic setting (S0).  False discovery proportions (A) and true positive rates (B) were averaged  over 100 simulation runs. Error bars (A) represent the  95\% CIs and the dashed horizontal line indicates the target FDR level of 0.05. }
\label{fig:sim:1}
\end{figure}

 \begin{figure}
\centering
\includegraphics[width=0.9\textwidth]{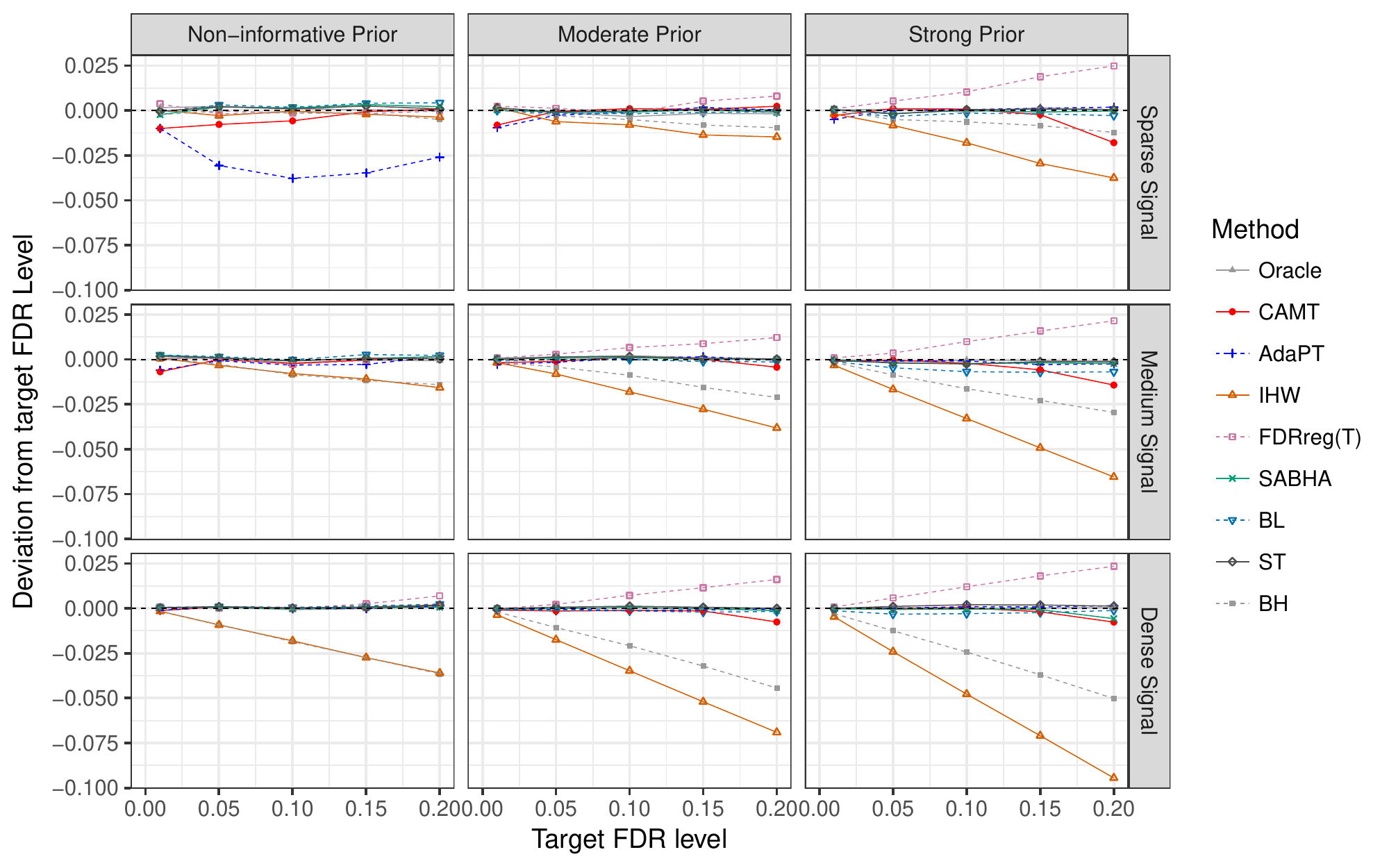}
\caption[FDR control  across different target levels. ]{FDR control at various target  levels (0.01 - 0.20) under the basic setting (S0) and a medium signal strength.  False discovery proportions  were averaged  over 100 simulation runs and the deviation from the target level (y-axis) was plotted.}
\label{fig:sim:2}
\end{figure}

 \begin{figure}
\centering
\includegraphics[width=0.75\textwidth]{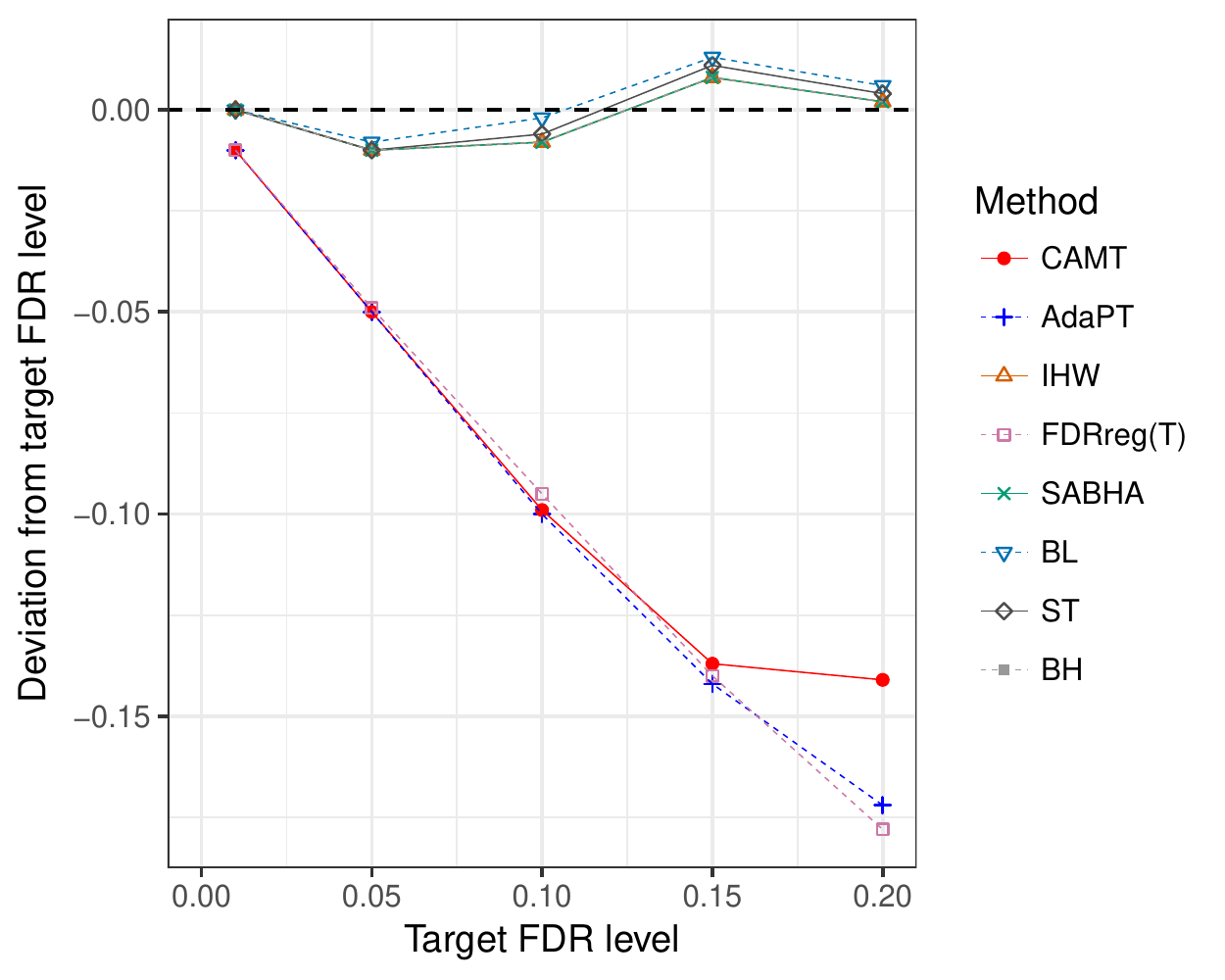}
\caption[FDR control under the complete null]{FDR  control at various target  levels (0.01 - 0.20) under the complete null (no signal was simulated). False discovery proportions were averaged over 1,000 simulation runs and the deviation from the target level (y-axis) was plotted. }
\label{fig:sim:3}
\end{figure}

  \begin{figure}
\centering
\includegraphics[width=0.9\textwidth]{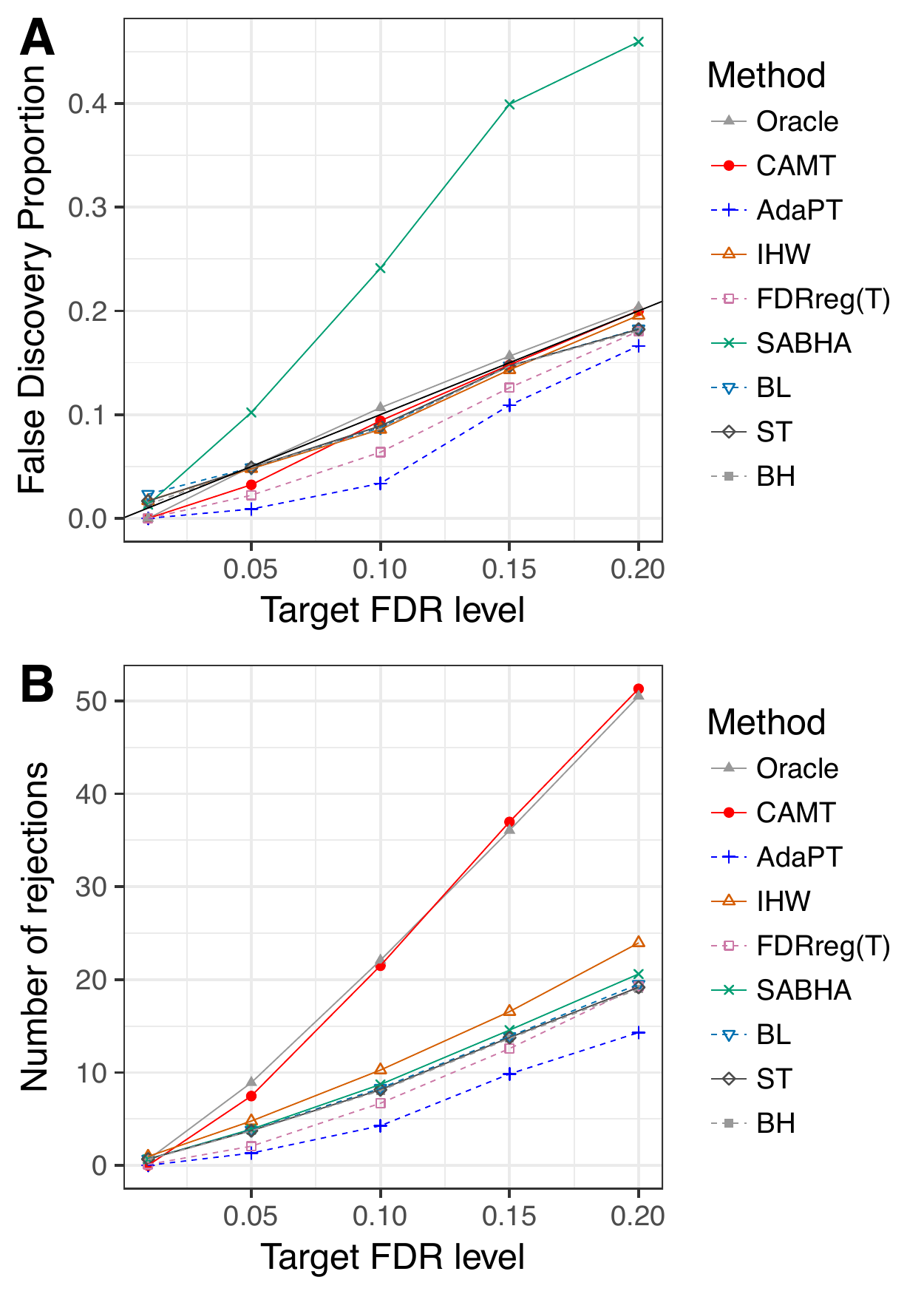}
\caption[Performance comparison under S0 and $m=100,000$]{Performance comparison with $m=100,000$ under the basic setting (S0).  Extremely low signal density ($>99\%$), moderate covariate strength and low signal strength were simulated. False discovery proportions (A) and number of rejections (B) were averaged  over 100 simulation runs and were plotted against various FDR target levels (0.01 - 0.20). }
\label{fig:sim:3:1}
\end{figure}

 \begin{figure}
\centering
\includegraphics[width=0.9\textwidth]{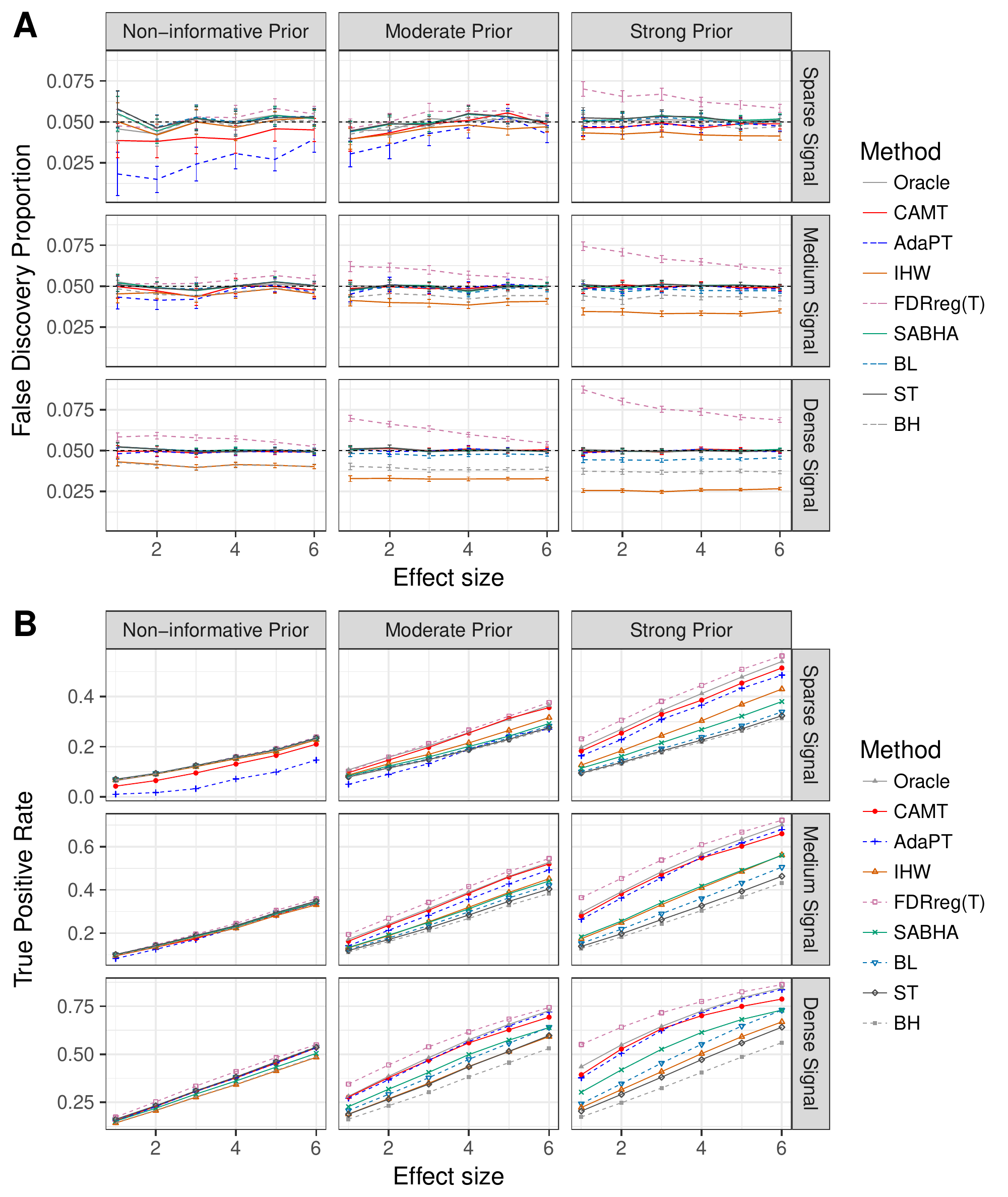}
\caption[Performance comparison under S2]{Performance comparison under S1 ($f_{1,i}$: non-central gamma distribution).  False discovery proportions (A) and true positive rates (B) were averaged  over 100 simulation runs. Error bars (A) represent the  95\% CIs and the dashed horizontal line indicates the target FDR level of 0.05. }
\label{fig:sim:4}
\end{figure}

 \begin{figure}
\centering
\includegraphics[width=0.9\textwidth]{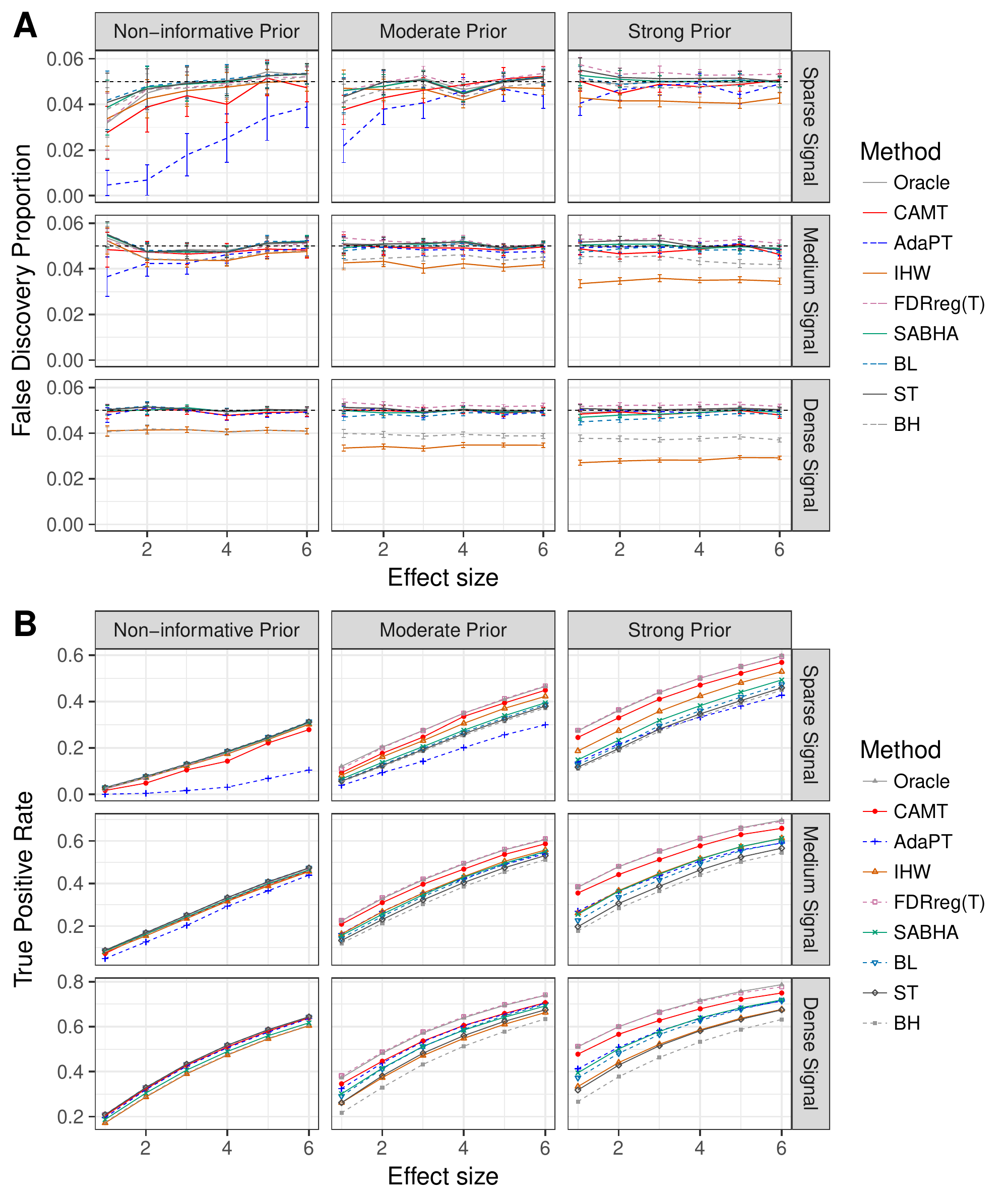}
\caption[Performance comparison under S1]{Performance comparison  under S2 (covariate-dependent $\pi_{0, i}$ and $f_{1, i}$).  False discovery proportions (A) and true positive rates (B) were averaged  over 100 simulation runs. Error bars (A) represent the  95\% CIs and the dashed horizontal line indicates the target FDR level of 0.05. }
\label{fig:sim:5}
\end{figure}

 \begin{figure}
\centering
\includegraphics[width=0.6\textwidth]{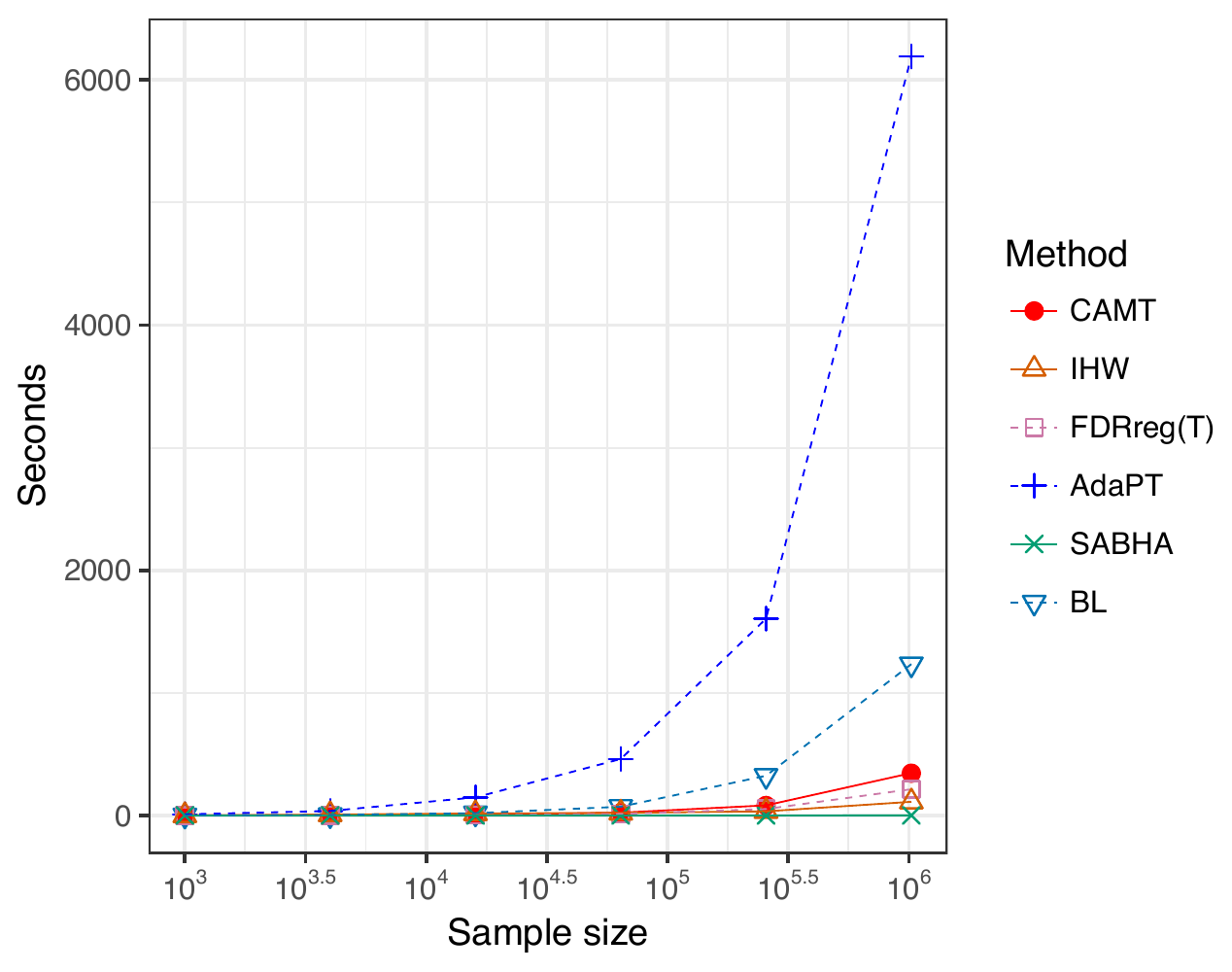}
\caption[Runtime comparison]{Comparison of runtime under the basic setting (S0).  Medium signal density and strength, and a moderately informative covariate was simulated.  The number of features varied from $10^3$ to $4^5 \times 10^3$.  The  average runtime over three replications was plotted against the feature size on a log scale. The computation was  performed on an AMD Opteron CPU with 256GB RAM and 16 MB available cache.}
\label{fig:sim:14}
\end{figure}

 \begin{figure}
\centering
\includegraphics[width=0.9\textwidth]{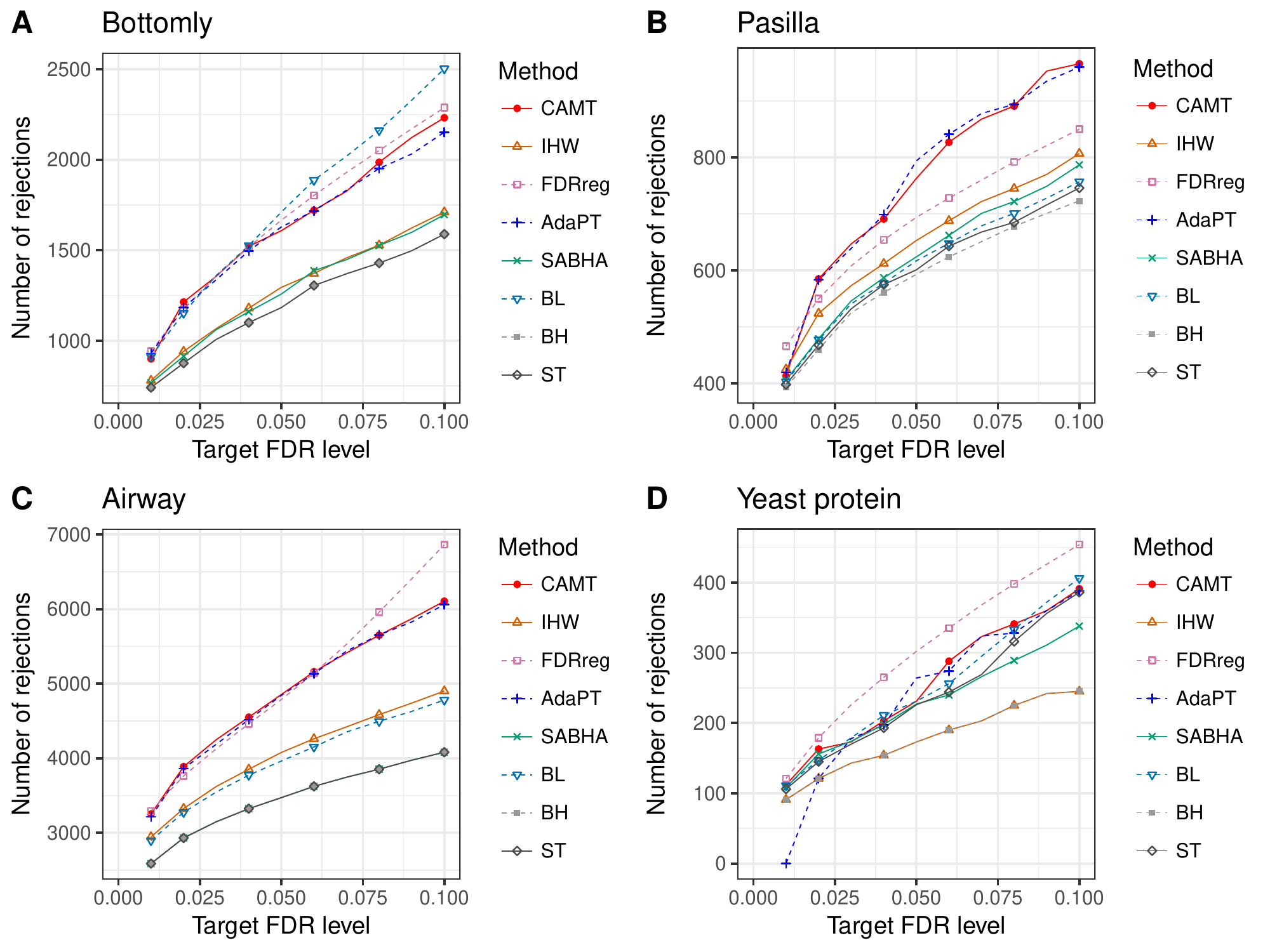}
\caption[Performance on the real datasets by AdaPT.]{The number of rejections at different target FDR levels on four real datasets used to demonstrate the performance of AdaPT. The Bottomly (A), Pasilla (B) and Airway (C) datasets were three transcriptomics datasets from RAN-seq experiments with a feature size of 13,932, 11,836 and 33,469,  respectively. The yeast protein dataset (D) was a proteomics dataset with a feature size of 2,666. }
\label{fig:real:1}
\end{figure}

 \begin{figure}
\centering
\includegraphics[width=0.9\textwidth]{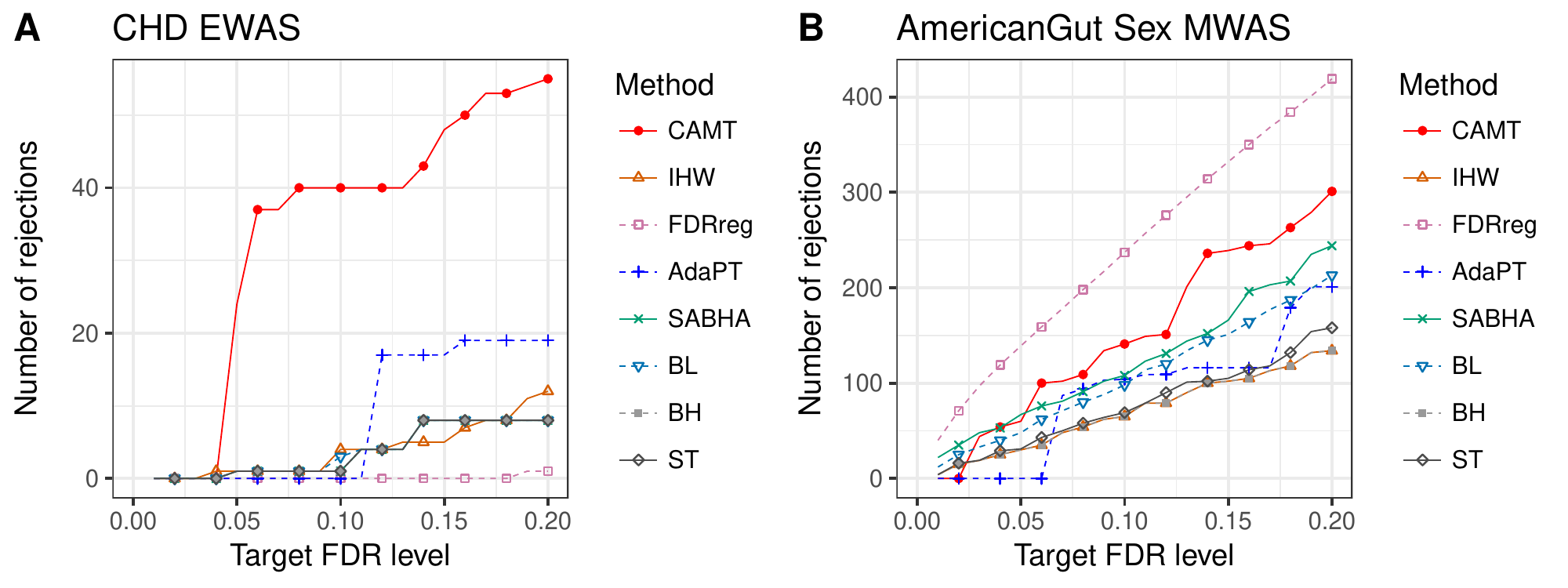}
\caption[Performance on the real datasets by AdaPT.]{The number of rejections at different target FDR levels on two real datasets: EWAS of congenital heart disease (A) and MWAS of sex effect (B). The EWAS dataset was produced by the Illumina 450K methylation beadchip ($m{=}455,741$) and the MWAS dataset was produced by the 16S rRNA gene amplicon sequencing  ($m{=}2,492$). }
\label{fig:real:2}
\end{figure}

\section{Application to omics-wide multiple testing}

To demonstrate the use of the proposed method for real-world applications, we applied CAMT to several omics datasets from transcriptomics, proteomics, epigenomics and metagenomics studies with the aim to identify omics features associated with the phenotype of interest. Since AdaPT is the most start-of-the-art method, we focused our comparison to it.  To make a fair comparison, we first run the analyses on the four omics datasets, which were also evaluated by AdaPT (Lei and Fithian, 2018), including Bottomly (Bottomly et al., 2011), Pasilla (Brooks et al., 2011), Airway (Himes et al., 2014) and Yeast Protein dataset (Dephoure et al., 2012). The Bottomly, Pasilla and Airway were three transcriptomics datasets from RNA-seq experiments with a feature size of 13,932, 11,836 and 33,469, respectively. The yeast protein dataset was a proteomics dataset from  with a feature size of 2,666. We used the same methods to calculate the  p-values for these datasets as described in Lei and Fithian (2018). The distributions of the p-values for these four datasets  all exhibited a spike in  the low p-value region, indicating that the signal was dense.  The logarithm of normalized count (averaged across all samples) was used as the univariate covariate for the three RNA-seq data (Bottomly, Pasilla and Airway). The logarithm of the total number of peptides across all samples was used as the univariate covariate for the yeast protein  data. Following  AdaPT, we used a spline basis with six equiquantile knots for $\pi_{0i}, f_{1,i}$ (CAMT and AdaPT) and for $\pi_{0i}$ (FDRreg, BL)  to account for potential complex nonlinear effects.   Since IHW and SABHA could only take univariate covariate, we used the univariate covariate directly.   We summarized the results in Figure \ref{fig:real:1}. We were able to reproduce the results in Lei and Fithian (2018). Indeed, AdaPT was  more powerful than SABHA, IHW, ST and BH on the four datasets.  { FDRreg and BL, which were not compared in Lei and Fithian (2018),  also performed well and made more rejections than other methods on the Yeast dataset and the Bottomly dataset, respectively.}  The performance of the proposed method, CAMT, was almost identical to AdaPT, which was consistent with the simulation results in the scenario of dense signal and informative covariate (Figure \ref{fig:sim:1}).

We next applied to two additional omics datasets from an epigenome-wide association study (EWAS) of congenital heart disease (CHD) (Wijnands et al., 2017) and a microbiome-wide association study (MWAS) of sex effect (McDonald et al., 2018).
\begin{itemize}
\item \textit{EWAS data}. The aim of the  EWAS of CHD was to identify the CpG loci in the human genome that were differentially methylated between healthy ($n=196$) and CHD ($n=84$) children. The methylation levels of 455,741 CpGs were measured by the the Illumina 450K methylation beadchip and was normalized properly before analysis.   The p-values were produced by running a linear regression to the methylation outcome for each CpG, adjusting for potential confounders such as age, sex and blood cell mixtures as described in Wijnands et al. (2017). Since widespread hyper-methylation (increased methylation in low-methylation regions) or hypo-methylation (decreased methylation in high-methylation regions) are common in many diseases (Robertson, 2005),  we use the mean methylation across samples as the univariate covariate.

\item \textit{MWAS data}. The aim of the MWAS of sex was to identify differentially abundant bacteria in the gut microbiome between males and females, where the abundances of  the gut bacteria were determined by sequencing a fingerprint gene in the bacteria 16S rRNA gene.  We used the publicly available data from the AmericanGut project (McDonald et al., 2018), where more than the gut microbiome from more than 10,000 subjects were sequenced. We focused our analysis on a relatively homogenous subset consisting of 481 males and 335 males (age between 13-70, normal BMI, from United States). We removed  OTUs (clustered sequencing units representing bacteria species) observed in less than 5 subjects, and a total of 2, 492 OTUs were tested using Wilcoxon rank sum test on the normalized abundances.   We use the percentage of zeros across samples as the univariate covariate since we expect a much lower power for OTUs with excessive zeros.
\end{itemize}

The results for these two datasets were summarized in Figure \ref{fig:real:2}.  For the EWAS data, the signal density was very sparse ($\hat{\pi}_0{= }0.99$, \textit{qvalue} package).  CAMT identified far more loci than the other methods at various FDR levels.   { The performance was consistent with the simulation results in the scenario of extremely sparse signal and informative covariate, where  CAMT was substantially more powerful than the competing methods (Figure \ref{fig:sim:3:1}). }At an FDR of 20\%, we identified 55 differentially methylated CpGs, compared to 19 for AdaPT.  These 55 CpG loci were mainly located in CpG islands  and the gene promotor regions, which were known for their important role in gene expression regulation (Robertson, 2005).  Interestingly, all but one CpG loci had low levels of methylation, indicating the methylation level was indeed informative to help identify differential CpGs.  We also did gene set enrichment analysis for the genes where the identified CpGs were located (https://david.ncifcrf.gov/). Based on the GO terms annotated to biological processes (BP\_DIRECT), three GO terms were found to be significant (unadjusted p-value $<$0.05)  including one term ``embryonic heart tube development", which was very relevant to the congenital heart disease under study (Wijnands et al., 2017).  As a sanity check, we randomized the covariate and re-analyzed the data using CAMT.   As expected, CAMT became similar to BH/ST and identified the same eight CpGs  at 20\% FDR level.

 For the MWAS data,  although the difference was not as striking as the EWAS data, CAMT was still overall more powerful than other competing methods except FDRreg. { However, given the fact that FDRreg was not robust under certain scenarios, the interpretation of the increased power should be cautious. } The relationship between the fitted $\pi_{0i}$  and the covariate (number of nonzeros) was very interesting:   $\hat{\pi}_{0i}$ first decreased, reached a minimum at around 70 and then increased (Figure \ref{fig:real:3}).  When the OTU was rare (e.g., a small number of nonzeros, only a few subjects had it), it was either very individualized or we had limited power to reject it, leading to a large $\pi_{0i}$.  In the other extreme where the OTU was very prevalent (e.g., a large number of nonzeros, most of the subjects had it), it was probability not sex-specific either.  Therefore, taking into account the sparsity level could increase the power of MWAS.   It is also informative  to compare CAMT to the traditional filtering-based procedure for MWAS.  In practice,  we usually apply a prevalence-based filter before performing multiple testing correction, based on the idea that rare OTUs are less likely to be significant and including them will increase the multiple testing burden.  A subjective filtering criterion has to be determined beforehand.  For this MWAS dataset,  if we removed OTUs present in less than 10\% of the subjects, ST and BH recovered 116 and 85  significant OTUs at an FDR of 10\%, compared to 69 and 65 on the original dataset, indicating that filtering did improve the statistical power of traditional FDR control procedures. However, if we removed OTUs present in less than 20\% of the subjects, the numbers of significant OTUs by  ST and BH reduced to 71 and 50 respectively. Therefore,  filtering could potentially leave out biologically important OTUs. In contrast,  CAMT did not require an explicit filtering criterion, and was much more powerful (141 significant OTUs at 10\% FDR) than the filtering-based method.

 \begin{figure}
\centering
\includegraphics[width=0.9\textwidth]{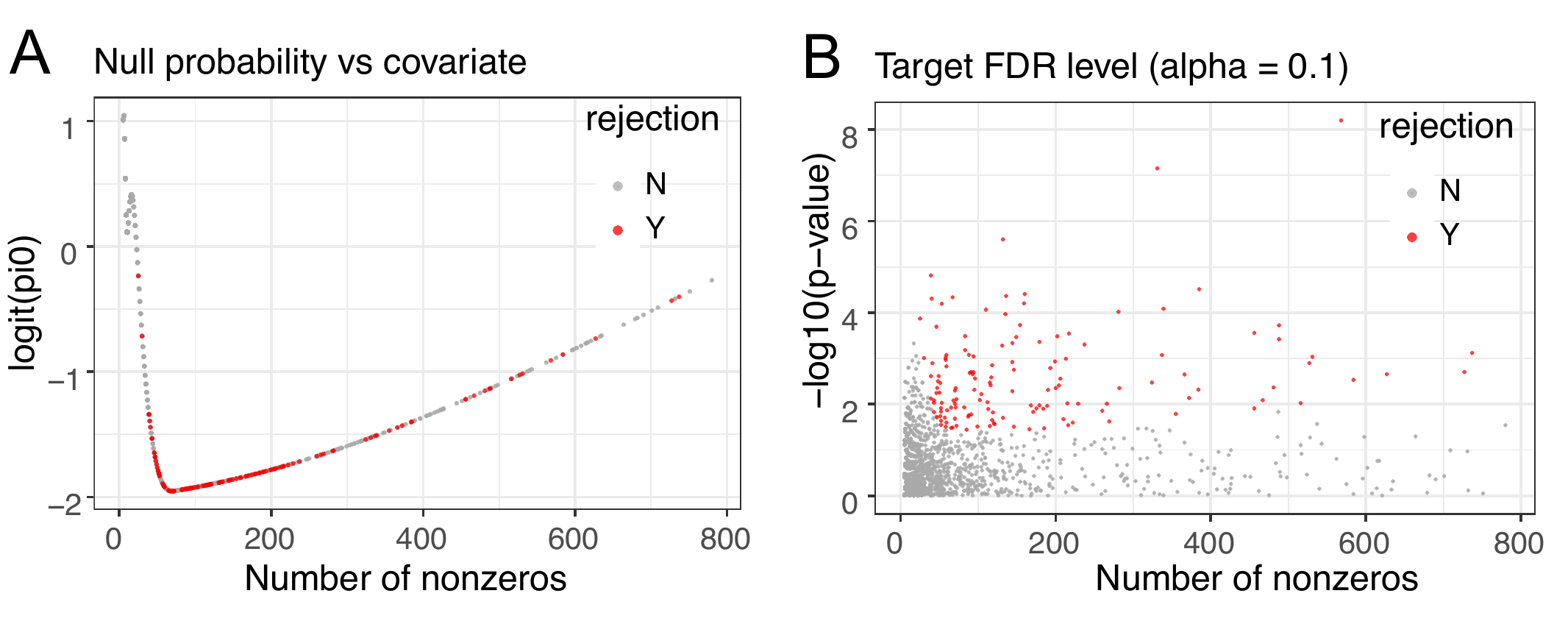}
\caption[Performance on MWAS data.]{Performance on the MWAS dataset. (A) The fitted $\pi_{0i}$  (logit scale) vs. the covariate (number of nonzeros). (B)  p-value (log scale) vs.  the covariate (number of nonzeros).  Rejected hypotheses at FDR 10\% were in red. }
\label{fig:real:3}
\end{figure}


\section{Discussions}\label{sec:dis}
There are generally two strategies for estimating the number of false rejections $\sum^{m}_{i=1}(1-H_i)\mathbf{1}\{h_i(p_i)\geq w_i(t)\}$
given the form of the rejection rule $h_i(p_i)\geq w_i(t)$. The first approach (called BH-type estimator) is to replace the number of false rejections by its expectation assuming that $p_i$
follows the uniform distribution on $[0,1]$ under the null, which leads to the quantity $\sum^{m}_{i=1}\pi_{0i}c(t,\pi_{0i},k_i)$ for $c(\cdot)$
defined in Section \ref{sec:asy}. The second approach (called BC-type estimator) estimates the false rejection conservatively by
$\xi+\sum^{m}_{i=1}\mathbf{1}\{h_i(1-p_i)\geq w_i(t)\}$ for a nonnegative constant $\xi$ under the assumption that the null distribution of p-values is symmetric about 0.5.
Both procedures enjoy optimality in some asymptotic sense, see, e.g., Arias-Castro and Chen (2017). The advantage of the BC-type procedure lies on that its estimation of the number of false rejections is asymptotically conservative when the rejection rule converges to a non-random limit (which holds even under a misspcified model, see e.g., White, 1982) and $f_0$ is mirror conservative (see equation (3) of Lei and Fithian, 2018).
This fact allows us to estimate the rejection rule by maximizing a potentially misspecified likelihood as the resulting rejection rule has a non-random limit under suitable conditions. This is not necessarily the case for the BH-type estimator without imposing additional constraint when estimating $\pi_{0i}$ and $k_i$. Specific restriction on the estimators of $\pi_{0i}$ is required for the BH-type estimator to achieve FDR control, see, e.g., equation (3) of Li and Barber (2018).

On the other hand, as the BC-type estimator uses a counting approach to estimate the number of false rejections, it suffers from
the discretization issue (i.e., the BC-type estimator is a step function of $t$ while the BH-type estimator is continuous), which may result in a large variance for the FDR estimate. This is especially the case when the FDR level is small.
For small FDR level, the number of rejections is usually small, and thus both the denominator and numerator of the FDR estimate become small and more variable.
Another issue with the BC-type estimator is the selection of $\xi$. We follow the idea of $\text{knockoff}+$ in Barber and Cand\`{e}s (2015)
by setting $\xi=1$. This choice could make the procedure rather conservative when the signal is very sparse, and the target FDR level is small.
A choice of smaller $\xi$ (e.g. $\xi=0$) often leads to inflated FDR in our unreported simulation studies. To alleviate this issue, one may
consider a mixed strategy by using
\begin{align*}
\max\left\{\sum^{m}_{i=1}\pi_{0i}c(t,\pi_{0i},k_i),\sum^{m}_{i=1}\mathbf{1}\{h_i(1-p_i)\geq w_i(t)\}\right\}
\end{align*}
as a conservative estimate for the number of false rejections when $t$ is relatively small. Our numerical results in Figure \ref{fig:sim:mix} in the supplementary material show that the resulting method can successfully
reduce the power loss in the case of sparse signals (or small FDR levels)  and less informative covariates while maintaining the good power performance in other cases. A serious investigation of this mixed procedure and the BH-type estimator is left for future research.

{ Since our method is not robust to a decreasing $f_0$,  some diagnostics are needed before running CAMT.  To detect a decreasing $f_0$,  the genomic inflation factor (GIF) can be employed (Devin and Roeder, 1999). GIF  is defined as the ratio of the median of the observed test statistic to the expected median based on the theoretical null distribution.  GIF has been widely used in genome-wide association studies to assess the deviation of the  empirical distribution of the null p-values from the theoretical uniform distribution.  To accommodate potential dense signals for some genomics studies, we recommend to confine the GIF calculation to p-values between 0.5 and 1.   If the GIF is larger, using CAMT  may result in excess false positives. In such case, the user should not trust the results and may consider recalculating the p-values by adjusting potential confounding factors,  either known or estimated based on some latent variable approach such as surrogate variable analysis (Leek and Storey, 2007), or using the simple genomic control approach based on p-values (Devin and Roeder, 1999). }

\newpage
\vskip 1em \centerline{\Large Supplement to ``Covariate Adaptive False Discovery Rate Control}
\centerline{\Large with Applications to Omics-Wide Multiple Testing"} 

\vskip 1em
\centerline{Xianyang Zhang and Jun Chen}
\setcounter{subsection}{0}
\renewcommand{\thesubsection}{A.\arabic{subsection}}
\setcounter{equation}{0}
\renewcommand{\theequation}{A.\arabic{equation}}

\doublespacing

\setcounter{table}{0}
\renewcommand{\thetable}{A\arabic{table}}%
\setcounter{figure}{0}
\renewcommand{\thefigure}{A\arabic{figure}}%
\setcounter{equation}{0}
\renewcommand{\theequation}{A\arabic{equation}}%
\setcounter{section}{0}
\renewcommand{\thesection}{A\arabic{section}}%

\section{Technical details}\label{sec:appendix}
\begin{proof}[Proof of Theorem \ref{thm-add}]
We first note that
\begin{align}
\text{FDR}=&E\left[\frac{\sum_{i=1}^m(1-H_i)\mathbf{1}\left\{h_i(p_i)\geq
w_i(t^*)\right\}}{1\vee\sum_{i=1
}^m\mathbf{1}\left\{h_i(p_i)\geq w_i(t^*)\right\}}\right] \nonumber
\\=&E\Bigg[\frac{1+\sum_{i=1}^m(1-H_i)\mathbf{1}\left\{h_i(1-p_i)\geq w_i(t^*)\right\}}{1\vee\sum_{i=1}^m\mathbf{1}\left\{h_i(p_i)\geq
w_i(t^*)\right\}} \nonumber
\\&\frac{\sum_{i=1}^m(1-H_i)\mathbf{1}\left\{h_i(p_i)\geq
w_i(t^*)\right\}}{1+\sum_{i=1}^m(1-H_i)\mathbf{1}\left\{h_i(1-p_i)\geq
w_i(t^*)\right\}}\Bigg] \nonumber
\\ \leq &\alpha E\left[\frac{\sum_{i=1}^m(1-H_i)\mathbf{1}\left\{h_i(p_i)\geq w_i(t^*)\right\}}{1+\sum_{s=1}^m(1-H_i)\mathbf{1}\left\{h_i(1-p_i)\geq
w_i(t^*)\right\}}\right] \nonumber
\\ =& \alpha E\left[\frac{\sum_{i=1}^m(1-H_i)\mathbf{1}\left\{\psi_i(p_i)\leq t^*\right\}}{1+\sum_{i=1}^m(1-H_i)\mathbf{1}\left\{ \psi_i(1-p_i)\leq
t^*\right\}}\right], \label{eq-thm3}
\end{align}
where
$$\psi_i(p_i)=\frac{\pi_{i}}{\pi_{i}+(1-\pi_{i})h_i(p_i)}.$$
As $h_i$ is strictly decreasing, $\psi_i(p)$ is a strictly
increasing function of $p.$ Define
$$b_i=\begin{cases}
\psi_i(p_i) & p_i<0.5, \\
\psi_i(1-p_i) & p_i\geq 0.5.
\end{cases}$$
Define the order statistics $b_{(1)}\leq b_{(2)}\leq \cdots \leq
b_{(m_0)}$ for $\{b_i: H_i=0\}$, where $m_0$ is the number of hypotheses under the null.
As $t^*\leq t_{\text{up}}$, we can find an integer
$J\leq m_0$ such that
$$b_{(1)}\leq b_{(2)}\leq \cdots \leq b_{(J)}\leq t^* < b_{(J+1)}\leq \cdots \leq b_{(m_0)}.$$
Then for any $J$,
\begin{align*}
\frac{\sum_{i=1}^m(1-H_i)\mathbf{1}\left\{\psi_i(p_i)\leq
t^*\right\}}{1+\sum_{i=1}^m(1-H_i)\mathbf{1}\left\{
\psi_i(1-p_i)\leq
t^*\right\}}=&\frac{(1-B_1)+\cdots+(1-B_J)}{1+B_1+B_2+\cdots+B_J}
=\frac{J+1}{1+B_1+B_2+\cdots+B_J}-1,
\end{align*}
where $B_i=\mathbf{1}\left\{p_{(i)}\geq 0.5\right\}.$ Under Condition (\ref{eq-con-f0}) of the main paper, we have
$\min_{i:H_i=0}P(p_i\geq 0.5)\geq 0.5.$ Using Lemma 1 of
Barber and Cand\`{e}s (2016), we have
$$E\left[\frac{J+1}{1+B_1+B_2+\cdots+B_J}-1\right]\leq 1.$$ By
(\ref{eq-thm3}), we have $\text{FDR} \leq \alpha$.
\end{proof}

\begin{proof}[Proof of Lemma \ref{lem-30}]
Let
$r_i(\theta,\beta)=\log\{\pi_{\theta}(x_i)+(1-
\pi_{\theta}(x_i))(1-k_\beta(x_i))p_i^{-k_\beta(x_i)}\}$ and $R_m(\theta,\beta)=\frac{1}{m}\sum^{m}_{i=1}r_i(\theta,\beta)$
for $\theta\in\Theta$ and $\beta\in\mathcal{B}$.
Under Assumptions \ref{ass-ad1} and the boundedness of $x_i$, there exist $c_1$ and $c_2$ such that
$0<c_1\leq \pi_{\theta}(x_i) \leq 1-c_1<1$ and $0<c_2\leq k_{\beta}(x_i) \leq 1-c_2<1$ for all $i$ and $\theta\in\Theta, \beta\in\mathcal{B}$. Thus we have
\begin{align*}
|r_i(\theta,\beta)|\leq& |\log\{\pi_{\theta}(x_i)+(1-
\pi_{\theta}(x_i))(1-c_2)p_i^{-(1-c_2)}\}|
\\ \leq&|\log\{2\pi_\theta(x_i)\}| +|\log\{2(1-\pi_\theta(x_i))(1-c_2)\}|+ (1-c_2)|\log(p_i)|
\\ \leq& c_3 + (1-c_2)|\log(p_i)|,
\end{align*}
for some constant $c_3>0$. Under Assumption \ref{ass-ad3}, by Corollary 1 of P\"{o}tscher and Prucha (1989), we
have
$$R_m(\theta,\beta)-E[R_m(\theta,\beta)]\rightarrow^{a.s.} 0$$
uniformly over $\theta\in\Theta$ and $\beta\in\mathcal{B}$. Together with Assumption \ref{ass-ad2}, we obtain
\begin{align*}
\sup_{\theta\in\Theta,\beta\in\mathcal{B}}|R_m(\theta,\beta)-R(\theta,\beta)|\rightarrow^{a.s.}0.
\end{align*}
Lemma 2.2 of White (1982) states
that if $(\hat{\theta},\hat{\beta})$ minimizes $-R_m$ and $(\theta^*,\beta^*)$ uniquely
minimizes $-R$, then
$$(\hat{\theta},\hat{\beta})\rightarrow^{a.s.} (\theta^*,\beta^*).$$
Finally, notice that
\begin{align*}
\max_{1\leq i\leq
m}|1/(1+e^{-\tilde{x}_i'\theta^*})-1/(1+e^{-\tilde{x}_i'\hat{\theta}})|=&\max_{1\leq i\leq
m}\frac{|e^{-\tilde{x}_i'\theta^*}-e^{-\tilde{x}_i'\hat{\theta}}|}{(1+e^{-\tilde{x}_i'\hat{\theta}})(1+e^{-\tilde{x}_i'\theta^*})}
\\ \leq& c_4 |\hat{\theta}-\theta^*|=o_{a.s.}(1),
\end{align*}
for $\tilde{x}_i=(1,x_i)'$, where the inequality follows from the
mean value theorem for $e^{-x}$ and the boundedness of $x_i$.
It implies that $\max_{1\leq i\leq
m}|\hat{\pi}_i-\pi_i^*|\rightarrow^{a.s.} 0$. The other result
$\max_{1\leq i\leq
m}|\hat{k}_i-\pi_{\beta^*}(x_i)|\rightarrow^{a.s.} 0$ follows from a similar argument.
\end{proof}

\begin{lemma}\label{lem-32}
Under Assumptions \ref{ass-ad3} and \ref{ass-32}, we have
\begin{align}
&\frac{1}{m}\sum_{i=1}^{m}\mathbf{1}\{p_i\leq
c(t,\pi_{i}^*,k^*_i)\}\rightarrow^{a.s.} G_0(t),\label{eq-sec2}\\
&\frac{1}{m}\sum_{i=1}^{m}\mathbf{1}\{1-p_i<
c(t,\pi_{i}^*,k^*_i)\}\rightarrow^{a.s.} G_1(t),\label{eq-sec0}\\
&\frac{1}{m}\sum_{H_i=0}\mathbf{1}\{p_i\leq
c(t,\pi_{i}^*,k^*_i)\}\rightarrow^{a.s.} \tilde{G}_1(t), \\
&\frac{1}{m}\sum_{i=1}^{m}\left\{\mathbf{1}\{p_i\leq t_i\}-P(p_i\leq
t_i)\right\}\rightarrow^{a.s.}0,\label{eq-sec3}
\end{align}
for any $t\geq t_0$ with $t_0>0$ and $t_i\in [0,1]$.
\end{lemma}
\begin{proof}[Proof of Lemma \ref{lem-32}]
Notice that
\begin{align*}
&\mathbf{1}\{p_i\leq c(t,\pi_i^*,k_i^*)\}=\mathbf{1}\left\{\psi_i(p_i,x_i)\leq t\right\},\\
&\mathbf{1}\{1-p_i<c(t,\pi_i^*,k_i^*)\}=\mathbf{1}\left\{\psi_i(1-p_i,x_i)< t\right\},
\end{align*}
where $\psi_i(p_i,x_i)=\pi_{\theta^*}(x_i)/\{\pi_{\theta^*}(x_i)+(1-\pi_{\theta^*}(x_i))(1-k_{\beta^*}(x_i))p_i^{-k_{\beta^*}(x_i)}\}$.
Under Assumption \ref{ass-ad3}, $\mathbf{1}\{\psi_i(p_i,x_i)\leq t\}$ and $\mathbf{1}\{\psi_i(1-p_i,x_i)<t\}$
are both $\alpha$-mixing (or $\phi$-mixing) processes. Thus by the strong law of large numbers for mixing processes, we get
\begin{align*}
&\frac{1}{m}\sum_{i=1}^{m}\left[\mathbf{1}\{p_i\leq
c(t,\pi_{i}^*,k^*_i)\}-P(p_i\leq
c(t,\pi_{i}^*,k^*_i))\right]\rightarrow^{a.s.} 0,\\
&\frac{1}{m}\sum_{i=1}^{m}\left[\mathbf{1}\{1-p_i<
c(t,\pi_{i}^*,k^*_i)\}-P(1-p_i<
c(t,\pi_{i}^*,k^*_i))\right]\rightarrow^{a.s.} 0,\\
&\frac{1}{m}\sum_{H_i=0}\left[\mathbf{1}\{p_i\leq
c(t,\pi_{i}^*,k^*_i)\}-P(p_i\leq
c(t,\pi_{i}^*,k^*_i))\right]\rightarrow^{a.s.} 0,
\end{align*}
and (\ref{eq-sec3}). The conclusion then follows from Assumption \ref{ass-32}.
\end{proof}

\begin{lemma}\label{lemma-10}
Suppose Assumptions \ref{ass-ad3}, \ref{ass-31} and \ref{ass-32} hold. Then for small
enough $\epsilon>0$, we have
\begin{align}
&\sup_{||K-K^*||_\infty<\epsilon}\sup_{||\Pi-\Pi^*||_{\infty}<\epsilon}\sup_{t\geq
t_0}\left|\frac{1}{m}\sum_{i=1}^{m}\mathbf{1}\{p_i\leq
c(t,\pi_i,k_i)\}-G_0(t)\right| \leq C_1\epsilon+o_{a.s.}(1),\label{eq-11}\\
&\sup_{||K-K^*||_\infty<\epsilon}\sup_{||\Pi-\Pi^*||_{\infty}<\epsilon}\sup_{t\geq
t_0}\left|\frac{1}{m}\sum_{i=1}^{m}\mathbf{1}\{1-p_i<
c(t,\pi_i,k_i)\}-G_1(t)\right| \leq
C_2\epsilon+o_{a.s.}(1), \label{eq-12}\\
&\sup_{||K-K^*||_\infty<\epsilon}\sup_{||\Pi-\Pi^*||_{\infty}<\epsilon}\sup_{t\geq
t_0}\left|\frac{1}{m}\sum_{i=1}^{m}(1-H_i)\mathbf{1}\{p_i\leq
c(t,\pi_i,k_i)\}-\tilde{G}_1(t)\right| \leq C_3\epsilon+o_{a.s.}(1),
\label{eq-13}
\end{align}
where $C_1,C_2,C_3>0$ are independent of $\epsilon,$ $||K-K^*||_{\infty}=\max_{1\leq i\leq m}|k_{i}-k_i^*|$,
$||\Pi-\Pi^*||_{\infty}=\max_{1\leq i\leq m}|\pi_{i}-\pi_i^*|$ and
$0<t_0<1.$
\end{lemma}
\begin{proof}[Proof of Lemma \ref{lemma-10}]
We only prove (\ref{eq-11}) as the proofs for the other results are
similar. For any $n$ with $1/n\leq \epsilon$, let
$q_{v,n}=G_0^{-1}(v/n)$ for $v=\lfloor n t_0 \rfloor,\dots,n$, where
$G_0^{-1}(x)=\inf\{u:G_0(u)\geq x\}$. Define
$$\hat{G}(t,\Pi,K)=\frac{1}{m}\sum_{i=1}^{m}\mathbf{1}\{p_i\leq
c(t,\pi_i,k_i)\}.$$ Note that $\hat{G}(t,\Pi,k)$ and $G_0(t)$ are both
non-decreasing functions of $t$. Denote by $\hat{G}(t-,\Pi,k)$ and
$G_0(t-)$ the left limits of $\hat{G}$ and $G_0$ at point $t$
respectively. Following the proof of Glivenko-Cantelli Lemma (see
Theorem 7.5.2 of Resnick, 2005), we have
\begin{align}
&\sup_{t\geq t_0}\left|\hat{G}(t,\Pi,K)-G_0(t)\right|  \nonumber
\\ \leq& \bigvee^{n}_{v=\lfloor n t_0 \rfloor}|\hat{G}(q_{v,n},\Pi,K)-G_0(q_{v,n})|\vee
|\hat{G}(q_{v,n}-,\Pi,K)-G_0(q_{v,n}-)|+1/n. \label{eq-res}
\end{align}
We note that for any $||K-K^*||_\infty<\epsilon$ and $||\Pi-\Pi^*||_{\infty}<\epsilon$,
\begin{align*}
& c(t,\pi_i,k_i)\leq c_+(t,\pi^*_i,k^*_i,\epsilon):=1\wedge \left\{\frac{t(1-k^*_i+\epsilon)(1-\pi_{i}^*+\epsilon)}{(1-t)(\pi_{i}^*-\epsilon)}\right\}^{1/(k^*_i+\epsilon)},\\
& c(t,\pi_i,k_i)\geq c_-(t,\pi^*_i,k^*_i,\epsilon):=1\wedge\left\{
\frac{t(1-k^*_i-\epsilon)(1-\pi_{i}^*-\epsilon)}{(1-t)(\pi_{i}^*+\epsilon)}\right\}^{1/(k^*_i-\epsilon)}.
\end{align*}
Also, note that $c_-(t,\pi^*_i,k^*_i,\epsilon)$ is bounded away from zero
for $t\geq \lfloor n t_0 \rfloor /n$ and large enough $n$. Define
\begin{align*}
&\hat{G}_+(t,\Pi^*,K^*,\epsilon)=\frac{1}{m}\sum_{i=1}^{m}\mathbf{1}\{p_i\leq
c_+(t,\pi_i^*,k^*_i,\epsilon)\},\\
&\hat{G}_-(t,\Pi^*,K^*,\epsilon)=\frac{1}{m}\sum_{i=1}^{m}\mathbf{1}\{p_i\leq
c_-(t,\pi_i^*,k^*_i,\epsilon)\}.
\end{align*}
We deduce that
\begin{align*}
&\sup_{||K-K^*||_\infty<\epsilon}\sup_{||\Pi-\Pi^*||_{\infty}<\epsilon}\sup_{t\geq
t_0}\left|\hat{G}(t,\Pi,K)-G_0(t)\right|
\\ \leq &\bigvee^{n}_{v=\lfloor n t_0 \rfloor}|\hat{G}_+(q_{v,n},\Pi^*,K^*,\epsilon)-G_0(q_{v,n})|\vee|\hat{G}_-(q_{v,n},\Pi^*,K^*,\epsilon)-G_0(q_{v,n})|
\\& \vee|\hat{G}_+(q_{v,n}-,\Pi^*,K^*,\epsilon)-G_0(q_{v,n}-)|\vee |\hat{G}_-(q_{v,n}-,\Pi^*,K^*,\epsilon)-G_0(q_{v,n}-)|+1/n.
\end{align*}
Next we analyze the term
$|\hat{G}_+(q_{v,n},\Pi^*,K^*,\epsilon)-G_0(q_{v,n})|$. As
$\pi_i^*$ and $k_i^*$ are both bounded away
from zero and one, some calculus shows
that
\begin{align*}
\max_{1\leq i\leq m}|c_+(q_{v,n},\pi_i^*,k^*_i,\epsilon)-
c(q_{v,n},\pi_i^*,k^*_i)|\leq c_1\epsilon.
\end{align*}
It thus implies that
\begin{equation}\label{eq-bd}
\begin{split}
&|P(p_i\leq c_+(q_{v,n},\pi_i^*,k^*_i,\epsilon))-P(p_i\leq c(q_{v,n},\pi_i^*,k^*_i))|
\\ \leq & E|P(p_i\leq c_+(q_{v,n},\pi_i^*,k^*_i,\epsilon)|x_i)-P(p_i\leq c(q_{v,n},\pi_i^*,k^*_i)|x_i)|
\\ \leq & E \max_{1\leq i\leq m}|c_+(q_{v,n},\pi_i^*,k^*_i,\epsilon)-
c(q_{v,n},\pi_i^*,k^*_i)| \leq c_1\epsilon.
\end{split}
\end{equation}
Then we have
\begin{align*}
&|\hat{G}_+(q_{v,n},\Pi^*,K^*,\epsilon)-G_0(q_{v,n})|
\\ \leq&
\left|\hat{G}_+(q_{v,n},\Pi^*,K^*,\epsilon)-\frac{1}{m}\sum_{i=1}^{m}P(p_i\leq
c_+(q_{v,n},\pi_i^*,k^*_i,\epsilon))\right|
\\&+\left|\frac{1}{m}\sum_{i=1}^{m}\{P(p_i\leq
c_+(q_{v,n},\pi_i^*,k^*_i,\epsilon))-P(p_i\leq
c(q_{v,n},\pi_i^*,k^*_i))\}\right|
\\&+\left|\frac{1}{m}\sum_{i=1}^{m}P(p_i\leq
c(q_{v,n},\pi_i^*,k^*_i))-G_0(q_{v,n})\right|
\\ \leq&
\left|\hat{G}_+(q_{v,n},\Pi^*,k^*_i,\epsilon)-\frac{1}{m}\sum_{i=1}^{m}P(p_i\leq
c_+(q_{v,n},\pi_i^*,k^*_i,\epsilon))\right|
\\&+c_0\max_{1\leq
i\leq m}|c_+(q_{v,n},\pi_i^*,k^*_i,\epsilon)-
c(q_{v,n},\pi_i^*,k^*_i)|+o(1)
\\ \leq& c_2\epsilon+o_{a.s.}(1),
\end{align*}
where the second inequality follows from Assumption \ref{ass-31},
(\ref{eq-sec2}) and (\ref{eq-sec3}), and the third inequality is
due to (\ref{eq-sec3}) and (\ref{eq-bd}). Similar arguments can be
used to deal with the other terms. Therefore,
$$\sup_{||K-K^*||_\infty<\epsilon}\sup_{||\Pi-\Pi^*||_{\infty}<\epsilon}\sup_{t\geq
t_0}\left|\hat{G}(t,\Pi,K)-G_0(t)\right|\leq
c_3\epsilon+1/n+o_{a.s}(1)\leq (c_3+1)\epsilon+o_{a.s.}(1).$$
\end{proof}

\begin{lemma}\label{lemma-11}
Under Assumptions \ref{ass-ad1}-\ref{ass-32}, we have
\begin{align}
&\sup_{t\geq t_0}\left|\frac{1}{m}\sum_{i=1}^{m}\mathbf{1}\{p_i\leq
c(t,\hat{\pi}_{i},\hat{k}_i)\}-G_0(t)\right|=o_{a.s.}(1),\label{eq2-hy}\\
&\sup_{t\geq t_0}\left|\frac{1}{m}\sum_{i=1}^{m}\mathbf{1}\{1-p_i<
c(t,\hat{\pi}_{i},\hat{k}_i)\}-G_1(t)\right|=o_{a.s.}(1),\label{eq3-hy}\\
&\sup_{t\geq
t_0}\left|\frac{1}{m}\sum_{i=1}^{m}(1-H_i)\mathbf{1}\{p_i\leq
c(t,\hat{\pi}_{i},\hat{k}_i)\}-\tilde{G}_1(t)\right|=o_{a.s.}(1).\label{eq11-hy}
\end{align}
\end{lemma}
\begin{proof}[Proof of Lemma \ref{lemma-11}]
To show (\ref{eq2-hy}), we define the event
$$\mathcal{A}_{m,\epsilon}=\{||\hat{\Pi}-\Pi^*||_{\infty}<\epsilon\}\cap
\{||\hat{K}-K^*||_\infty<\epsilon\}.$$ Conditional on
$\mathcal{A}_{m,\epsilon}$, by Lemma \ref{lemma-10} and Assumption
\ref{ass-31}, we have
\begin{align*}
&\sup_{t\geq
t_0}\left|\frac{1}{m}\sum_{i=1}^{m}\left(\mathbf{1}\{p_i\leq
c(t,\hat{\pi}_{i},\hat{k}_i)\}-\mathbf{1}\{p_i\leq
c(t,\pi_{i}^*,k^*_i)\}\right)\right|
\\ \leq & 2\sup_{||K-K^*||_\infty<\epsilon}\sup_{||\Pi-\Pi^*||_{\infty}<\epsilon}\sup_{t\geq t_0}\left|\frac{1}{m}\sum_{i=1}^{m}\left(\mathbf{1}\{p_i\leq  c(t,\pi_i,k_i)\}\right)-P(p_i\leq  c(t,\pi_i,k_i))\right|
\\&+\sup_{t\geq t_0}\left|\frac{1}{m}\sum_{i=1}^{m}\left(P(p_i\leq  c(t,\hat{\pi}_{i},\hat{k}))-P(p_i\leq  c(t,\pi_i^*,k^*))\right)\right|
\\ \leq & c_0\max_{1\leq i\leq
m}\sup_{t\geq t_0}\left| c(t,\hat{\pi}_{i},\hat{k}_i)-
c(t,\pi_{i}^*,k^*_i)\right|+c_4\epsilon+o_{a.s.}(1)
\\=&c_5\epsilon+o_{a.s.}(1).
\end{align*}
As $P(\mathcal{A}_{m,\epsilon})\rightarrow 1$, the conclusion
follows. The proofs for (\ref{eq3-hy}) and (\ref{eq11-hy}) are similar
and we omit the details.
\end{proof}

\begin{lemma}\label{lemma-12}
Under Assumptions \ref{ass-ad1}-\ref{ass-32}, we have
\begin{align}
&\sup_{t\geq t'}\left|\frac{\sum_{i=1}^{m}\mathbf{1}\{1-p_i<
c(t,\hat{\pi}_{i},\hat{k}_i)\}}{\sum_{i=1}^{m}\mathbf{1}\{p_i\leq
c(t,\hat{\pi_{i}},\hat{k}_i)\}}-\frac{G_1(t)}{G_0(t)}\right|=o_{a.s.}(1),\\
&\sup_{t\geq t'}\left|\frac{\sum_{i=1}^{m}(1-H_i)\mathbf{1}\{p_i\leq
c(t,\hat{\pi}_{i},\hat{k}_i)\}}{\sum_{i=1}^{m}\mathbf{1}\{p_i\leq
c(t,\hat{\pi}_{i},\hat{k}_i)\}}-\frac{\tilde{G}_1(t)}{G_0(t)}\right|=o_{a.s.}(1).
\end{align}
\end{lemma}
\begin{proof}[Proof of Lemma \ref{lemma-12}]
For the ease of presentation, denote
$m^{-1}\sum_{i=1}^{m}\mathbf{1}\{1-p_i<
c(t,\hat{\pi}_{i},\hat{k}_i)\}$ and
$m^{-1}\sum_{i=1}^{m}\mathbf{1}\{p_i\leq
c(t,\hat{\pi}_{i},\hat{k}_i)\}$ by $G_{m,1}$ and $G_{m,0}$
respectively. The monotonicity of $G_0$ implies that $\min_{t\geq
t'}G_0(t)=G_0(t')>0.$ By Lemma \ref{lemma-11}, we deduce that
\begin{align*}
&\left|\frac{G_{m,1}(t)}{G_{m,0}(t)}-\frac{G_1(t)}{G_0(t)}\right|
\\ =& \left|\frac{(G_{m,1}(t)-G_1(t))G_0(t)-G_1(t)(G_{m,0}(t)-G_0(t))}{G_0(t)G_{m,0}(t)}\right|
\\ \leq & \frac{G_0(1)|G_{m,1}(t)-G_1(t)|+G_1(1)|G_{m,0}(t)-G_0(t)|}{G_0(t')\{G_0(t)-\sup_{x\geq t'}|G_{m,0}(x)-G_{0}(x)|\}}
\\ \leq & \frac{G_0(1)\sup_{x\geq t'}|G_{m,1}(x)-G_1(x)|+G_1(1)\sup_{x\geq t'}|G_{m,0}(x)-G_0(x)|}{G_0(t')\{G_0(t')-\sup_{x\geq t'}|G_{m,0}(x)-G_{0}(x)|\}}
\rightarrow^{a.s.} 0,
\end{align*}
uniformly for any $t\geq t'$. Similar arguments can be used to prove
the other result.
\end{proof}

\begin{lemma}\label{lem-add}
Suppose $f_0$ satisfies Condition (\ref{eq-con-f0}) of the main paper. Under Assumption \ref{ass-32}, we have
$\tilde{G}_1(t)\leq G_1(t)$ for $t\geq t_0$.
\end{lemma}
\begin{proof}[Proof of Lemma \ref{lem-add}]
Under Assumption \ref{ass-32}, we have
\begin{align*}
\tilde{G}_1(t)=&\lim_{m\rightarrow+\infty}\frac{1}{m}\sum_{H_i=0}P(p_i\leq
c(t,\pi_{i}^*,k^*_i))
\\ \leq &\lim_{m\rightarrow+\infty}\frac{1}{m}\sum_{H_i=0}P(1-p_i<
c(t,\pi_{i}^*,k^*_i))
\\ \leq &\lim_{m\rightarrow+\infty}\frac{1}{m}\sum_{i=1}^{m}P(1-p_i<
c(t,\pi_{i}^*,k^*_i))=G_1(t)
\end{align*}
for $t\geq t_0$, where the first inequality is due to the assumption $P(p_i\leq a)\leq P(1-p_i\leq a)$ for any $a\in [0,1]$ and $p_i\sim f_0$.
\end{proof}

\begin{proof}[Proof of Theorem \ref{thm}]
Set $e=\alpha-U(t')$. By Lemma \ref{lemma-12}, we have
$$\frac{\sum_{i=1}^{m}\mathbf{1}\{1-p_i<
c(t',\hat{\pi}_{i},\hat{k})\}}{\sum_{i=1}^{m}\mathbf{1}\{p_i\leq
 c(t',\hat{\pi_{i}},\hat{k})\}}\leq \alpha-e/2<\alpha,$$ with
probability tending to one. Therefore, $P(\hat{t}\geq t')\rightarrow
1$ as $m\rightarrow+\infty.$ Then we have
\begin{align*}
& \alpha-\frac{\sum_{i=1}^{m}(1-H_i)\mathbf{1}\{p_i\leq
c(\hat{t},\hat{\pi}_{i},\hat{k}_i)\}}{\sum_{i=1}^{m}\mathbf{1}\{p_i\leq
c( \hat{t},\hat{\pi}_{i},\hat{k}_i)\}}
\\ \geq &\frac{\sum_{i=1}^{m}\mathbf{1}\{1-p_i<
c(\hat{t},\hat{\pi}_{i},\hat{k}_i)\}}{\sum_{i=1}^{m}\mathbf{1}\{p_i\leq
c(\hat{t},\hat{\pi_{i}},\hat{k}_i)\}}-\frac{\sum_{i=1}^{m}(1-H_i)\mathbf{1}\{p_i\leq
c(\hat{t},\hat{\pi}_{i},\hat{k}_i)\}}{\sum_{i=1}^{m}\mathbf{1}\{p_i\leq
c( \hat{t},\hat{\pi}_{i},\hat{k}_i)\}}
\\ \geq& \inf_{t\geq
t'}\Bigg\{\frac{\sum_{i=1}^{m}\mathbf{1}\{1-p_i<
c(t,\hat{\pi}_{i},\hat{k}_i)\}}{\sum_{i=1}^{m}\mathbf{1}\{p_i\leq
c(t,\hat{\pi_{i}},\hat{k}_i)\}}-\frac{G_1(t)}{G_0(t)}
\\&+\frac{G_1(t)-\tilde{G}_1(t)}{G_0(t)}+\frac{\tilde{G}_1(t)}{G_0(t)}-\frac{\sum_{i=1}^{m}(1-H_i)\mathbf{1}\{p_i\leq
c(t,\hat{\pi}_{i},\hat{k}_i)\}}{\sum_{i=1}^{m}\mathbf{1}\{p_i\leq
c(t,\hat{\pi}_{i},\hat{k}_i)\}}\Bigg\}\geq o_{a.s.}(1),
\end{align*}
where we have used the fact that $G_1(t)\geq \tilde{G}_1(t)$ as shown in Lemma \ref{lem-add}. It implies that
\begin{align*}
\frac{\sum_{i=1}^{m}(1-H_i)\mathbf{1}\{p_i\leq
c(\hat{t},\hat{\pi}_{i},\hat{k}_i)\}}{\sum_{i=1}^{m}\mathbf{1}\{p_i\leq
c(\hat{t},\hat{\pi}_{i},\hat{k}_i)\}}\leq \alpha+o_{a.s.}(1).
\end{align*}
Finally by Fatou's Lemma,
\begin{align*}
&\limsup_{m}\text{FDR}(\hat{t},\hat{\Pi},\hat{K})\leq
\limsup_{m}E\left[\frac{\sum_{i=1}^{m}(1-H_i)\mathbf{1}\{p_i\leq
c(\hat{t},\hat{\pi}_{i},\hat{k}_i)\}}{\sum_{i=1}^{m}\mathbf{1}\{p_i\leq
c(\hat{t},\hat{\pi}_{i},\hat{k}_i)\}} \right]\leq \alpha,
\end{align*}
which completes the proof.
\end{proof}

\section{Additional simulation results}
We summarize additional simulation results  in Figures \ref{fig:sim:6:1}-\ref{fig:sim:13}.

 \begin{figure}
\centering
\includegraphics[width=0.9\textwidth]{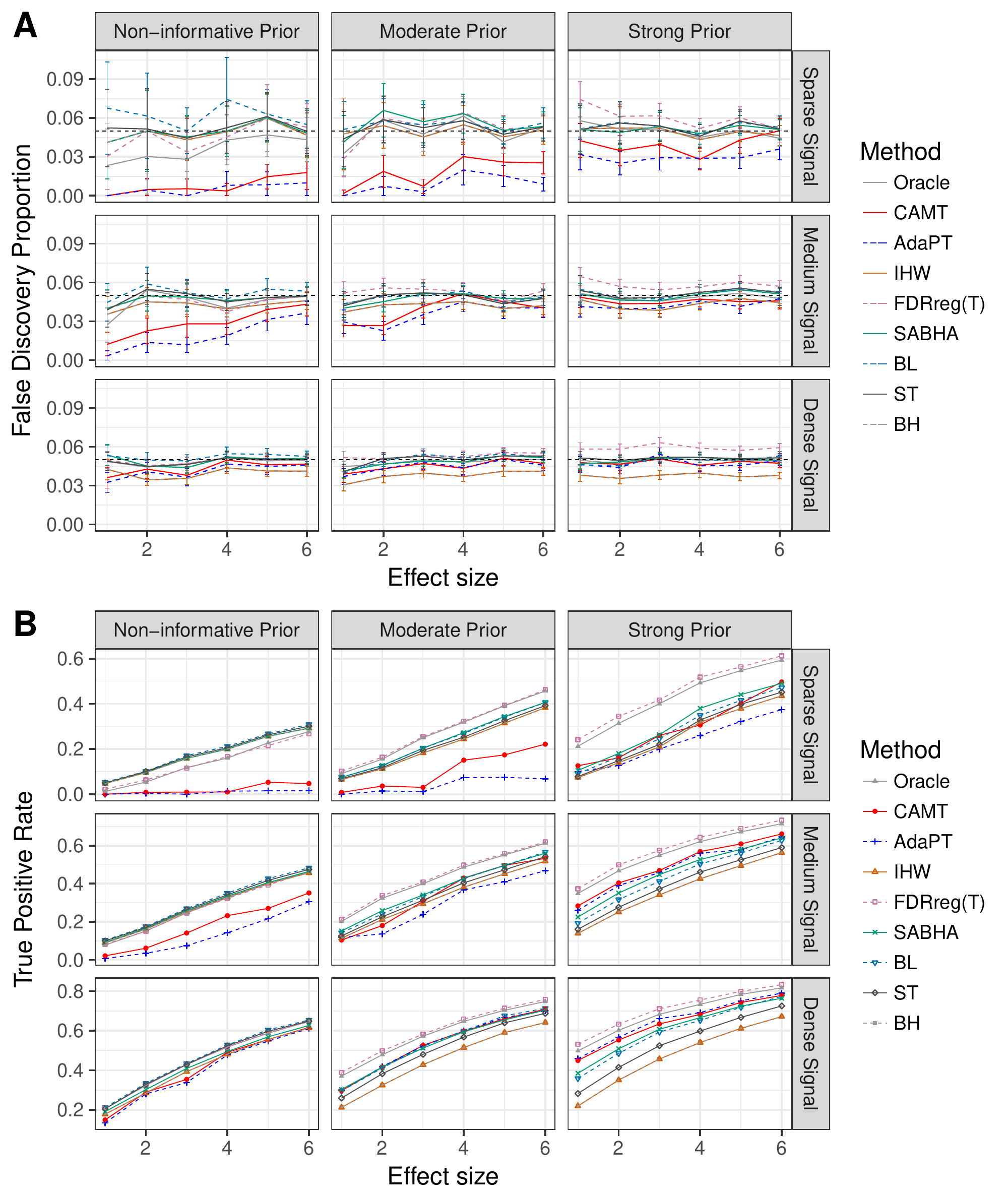}
\caption[Performance comparison under S0 and $m=1,000$]{Performance comparison with $m=1,000$ under the basic setting (S0).  False discovery proportions (A) and true positive rates (B) were averaged  over 100 simulation runs. Error bars (A) represent the  95\% CIs and the dashed horizontal line indicates the target FDR level of 0.05. }
\label{fig:sim:6:1}
\end{figure}

 \begin{figure}
\centering
\includegraphics[width=0.9\textwidth]{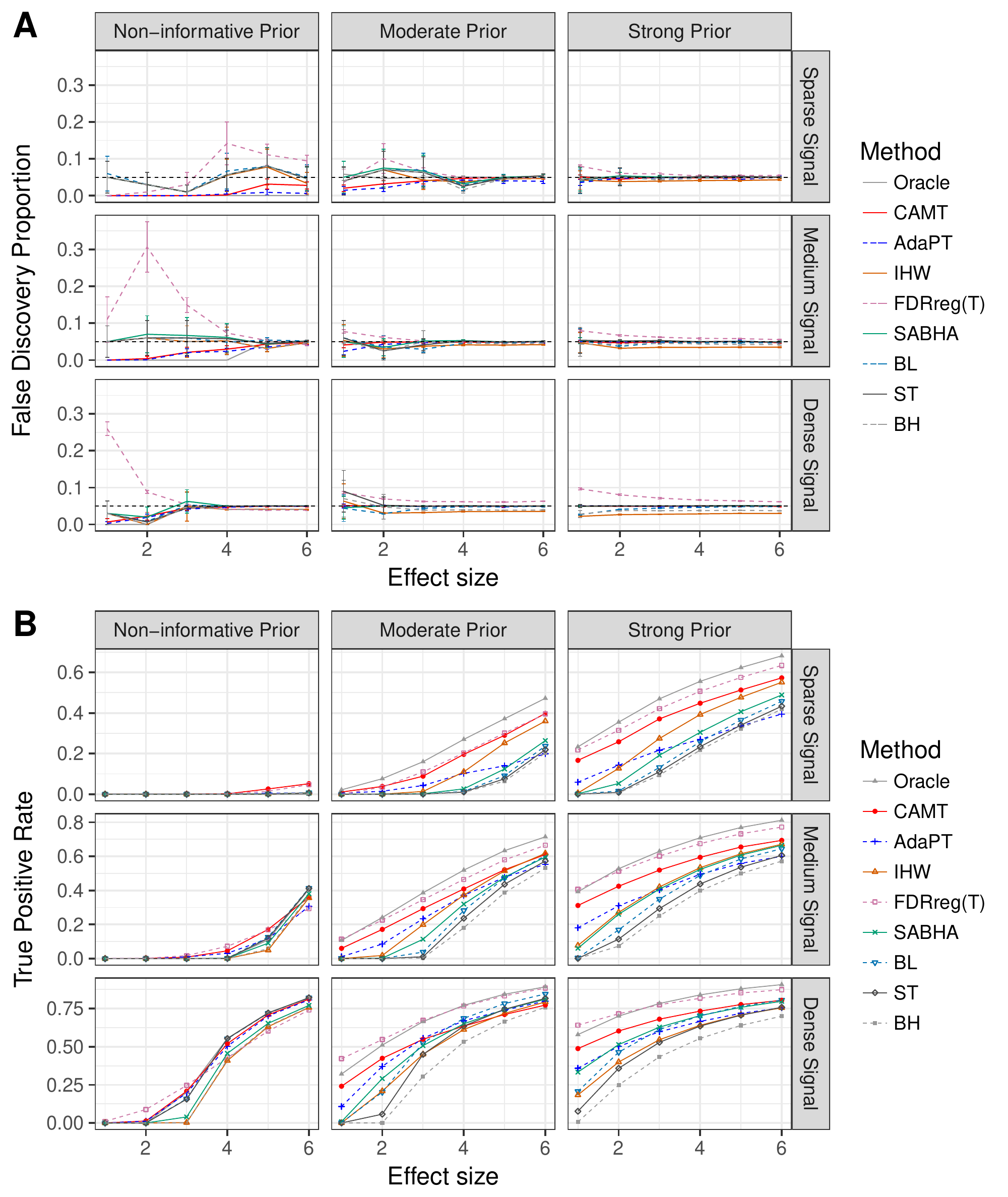}
\caption[Performance comparison under S2]{Performance comparison  under S2 (covariate-dependent $\pi_{0, i}$ and $f_{1, i}$, the standard deviation of the z-score under $H_1$ is 0.5).  False discovery proportions (A) and true positive rates (B) were averaged  over 100 simulation runs. Error bars (A) represent the  95\% CIs and the dashed horizontal line indicates the target FDR level of 0.05. }
\label{fig:sim:6}
\end{figure}

 \begin{figure}
\centering
\includegraphics[width=0.9\textwidth]{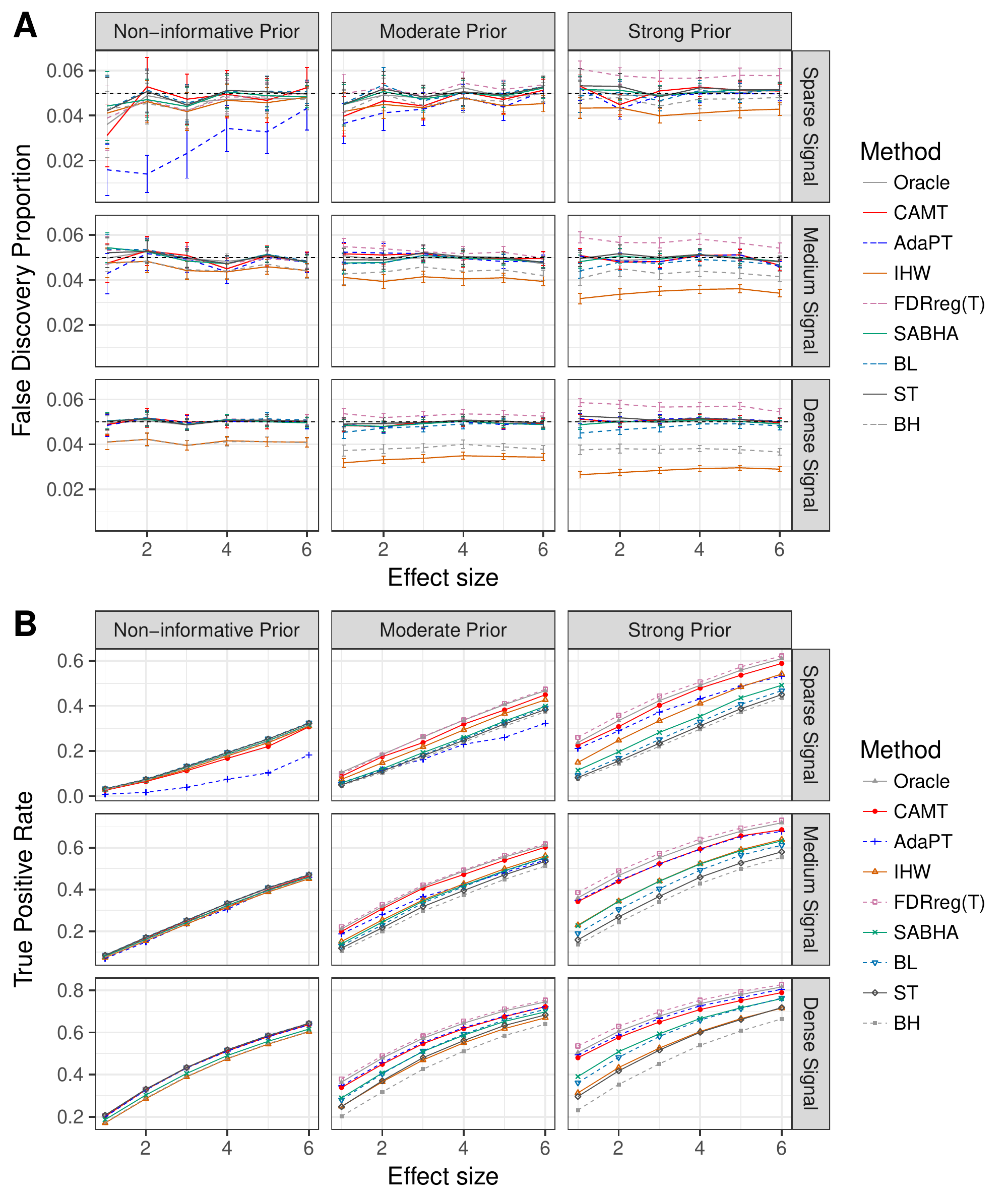}
\caption[Performance comparison under S3.1]{Performance comparison  under S3.1 (block correlation structure, positive correlations ($\rho{=}0.5$) within blocks).  False discovery proportions (A) and true positive rates (B) were averaged  over 100 simulation runs. Error bars (A) represent the  95\% CIs and the dashed horizontal line indicates the target FDR level of 0.05. }

\label{fig:sim:7}
\end{figure}

 \begin{figure}
\centering
\includegraphics[width=0.9\textwidth]{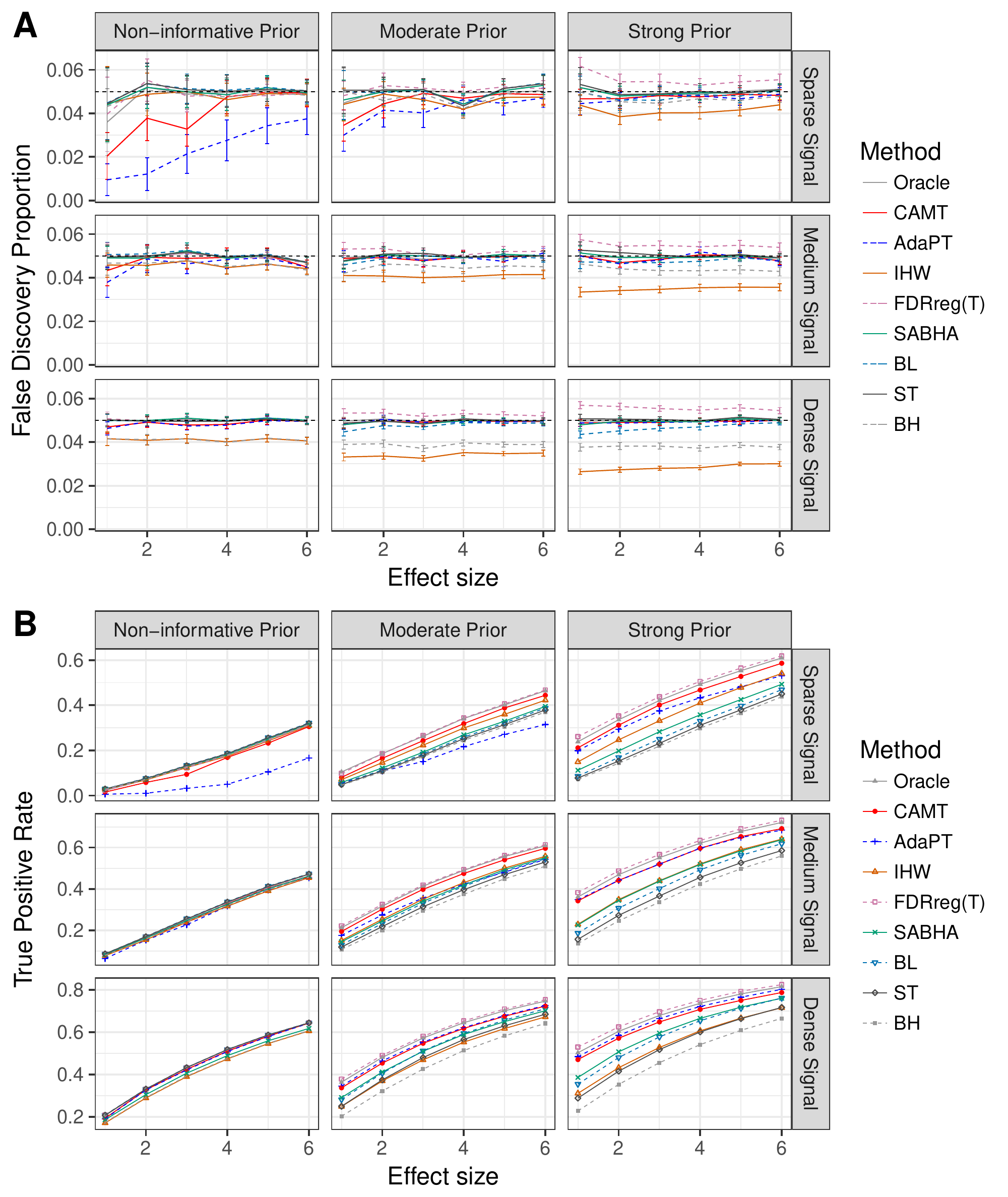}
\caption[Performance comparison under S3.2]{Performance comparison  under S3.2 (block correlation structure, positive/negative correlations ($\rho{=}\pm0.5$)  within blocks).  False discovery proportions (A) and true positive rates (B) were averaged  over 100 simulation runs. Error bars (A) represent the  95\% CIs and the dashed horizontal line indicates the target FDR level of 0.05. }

\label{fig:sim:8}
\end{figure}

 \begin{figure}
\centering
\includegraphics[width=0.9\textwidth]{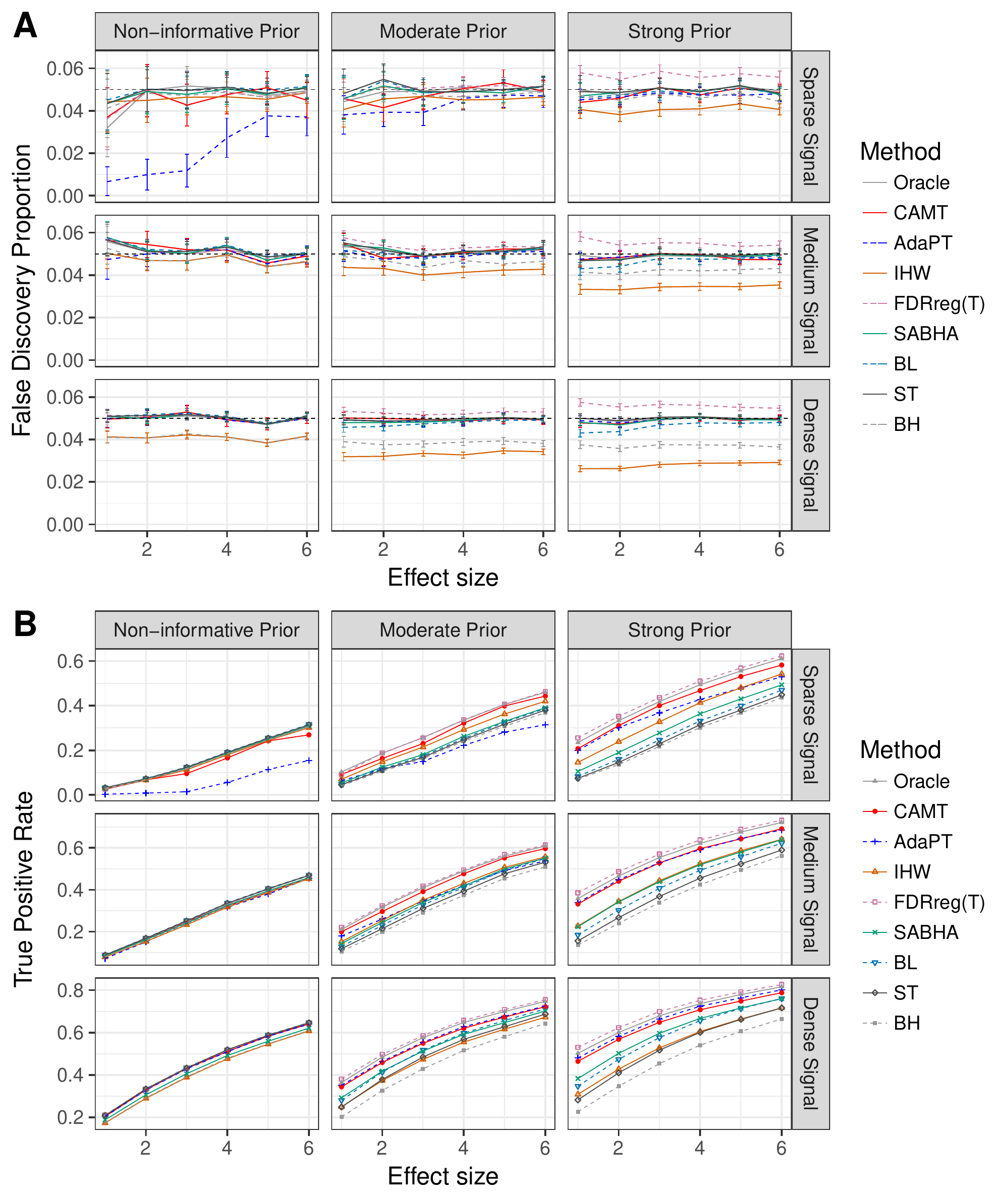}
\caption[Performance comparison under S3.3]{Performance comparison  under S3.3 (AR(1) structure, positive correlations ($\rho{=}0.75$)).  False discovery proportions (A) and true positive rates (B) were averaged  over 100 simulation runs. Error bars (A) represent the  95\% CIs and the dashed horizontal line indicates the target FDR level of 0.05. }
\label{fig:sim:9}
\end{figure}

 \begin{figure}
\centering
\includegraphics[width=0.9\textwidth]{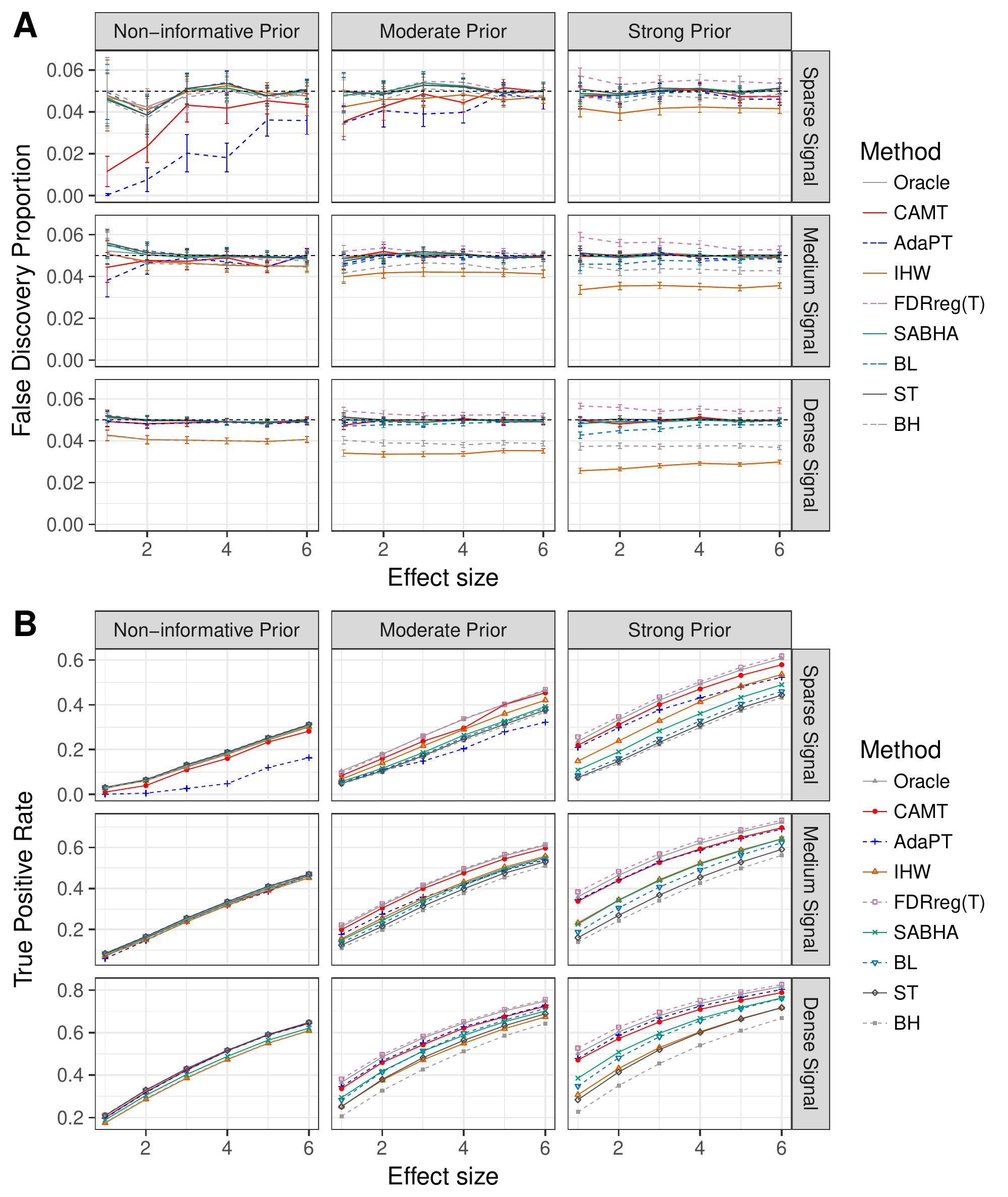}
\caption[Performance comparison under S3.4]{Performance comparison  under S3.4 (AR(1) structure, positive/negative correlations ($\rho{=}\pm0.75$)).  False discovery proportions (A) and true positive rates (B) were averaged  over 100 simulation runs. Error bars (A) represent the  95\% CIs and the dashed horizontal line indicates the target FDR level of 0.05. }
\label{fig:sim:10}
\end{figure}

\begin{figure}
\centering
\includegraphics[width=0.9\textwidth]{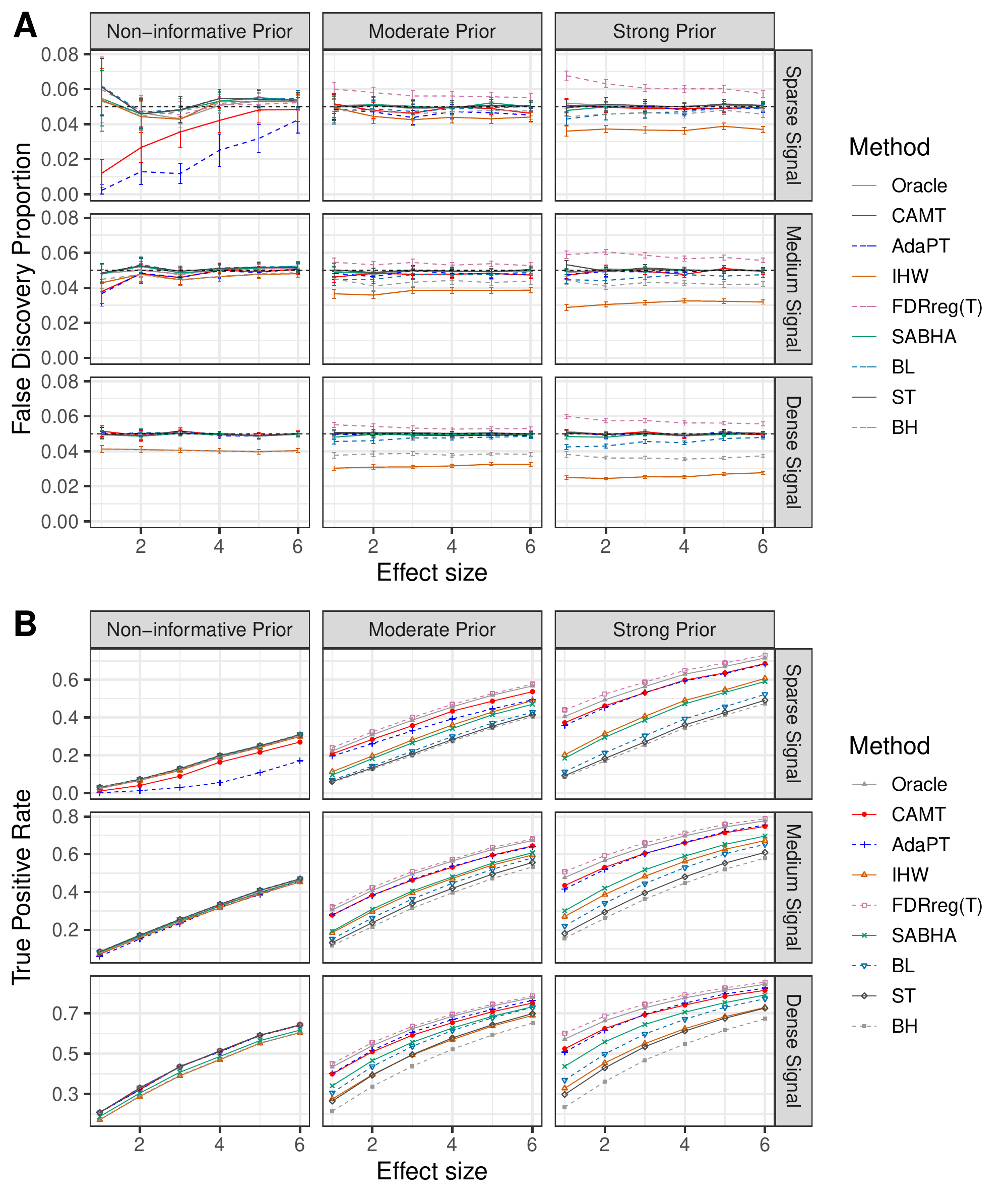}
\caption[Performance comparison under S4]{Performance comparison under S4 (heavy-tail covariate).  False discovery proportions (A) and true positive rates (B) were averaged  over 100 simulation runs. Error bars (A) represent the  95\% CIs and the dashed horizontal line indicates the target FDR level of 0.05. }
\label{fig:sim:11}
\end{figure}

 \begin{figure}
\centering
\includegraphics[width=0.9\textwidth]{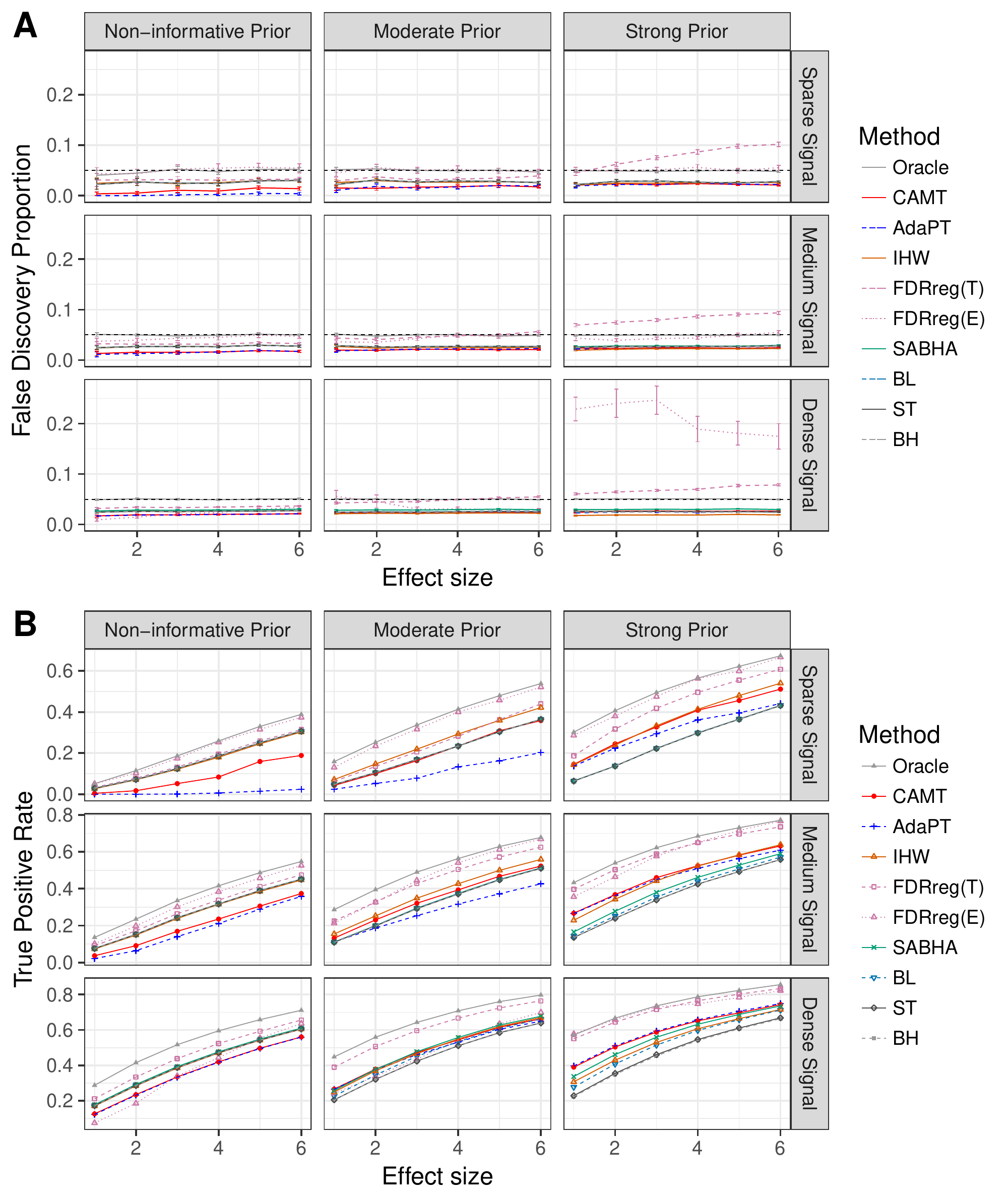}
\caption[Performance comparison under S5.1]{Performance comparison  under S5.1 (increasing $f_0$).  False discovery proportions (A) and true positive rates (B) were averaged  over 100 simulation runs. Error bars (A) represent the  95\% CIs and the dashed horizontal line indicates the target FDR level of 0.05. FDRregT and FDRregE represent the FDRreg method using the theoretical and empirical null respectively. }

\label{fig:sim:12}
\end{figure}

 \begin{figure}
\centering
\includegraphics[width=0.9\textwidth]{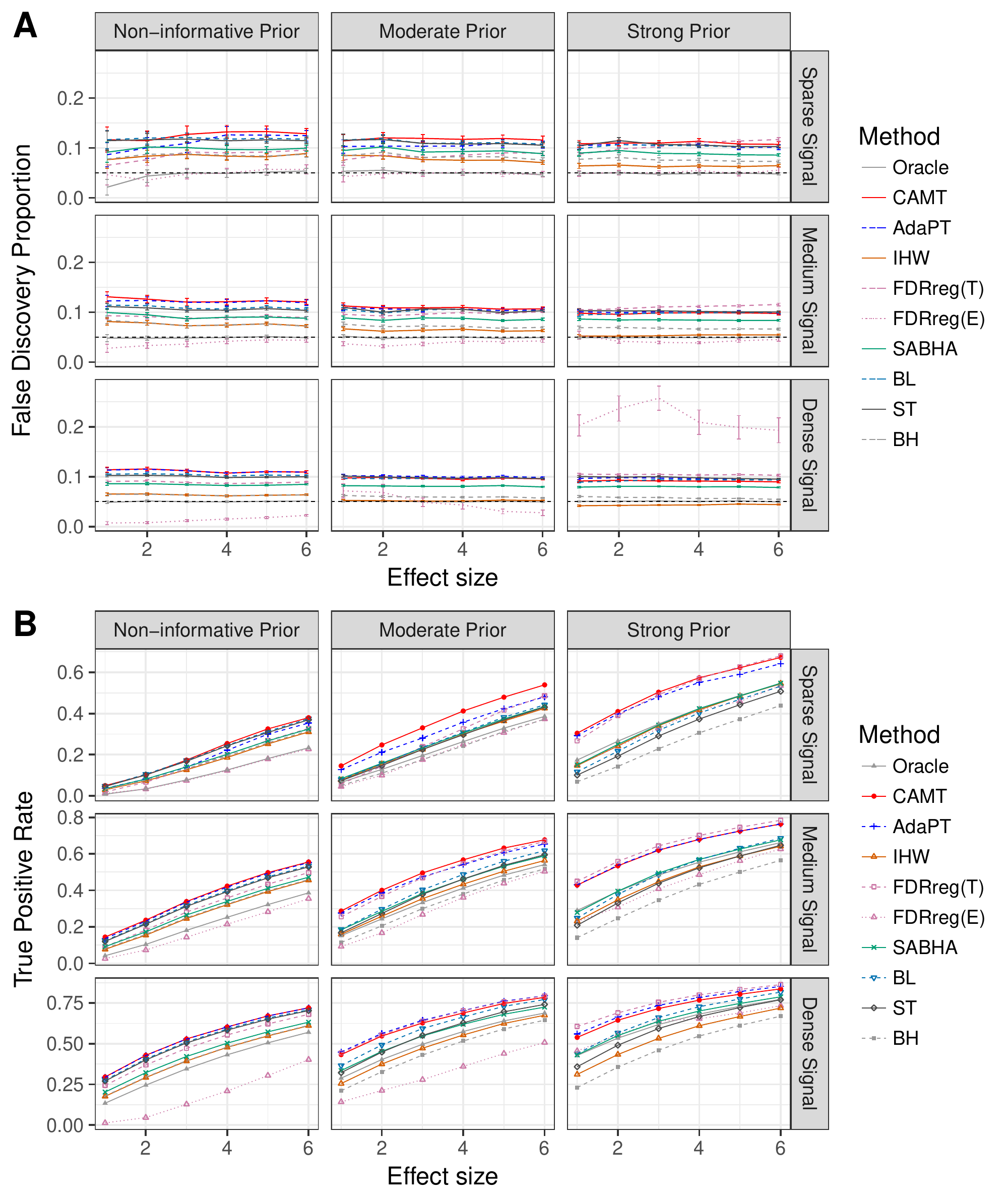}
\caption[Performance comparison under S5.2]{Performance comparison  under S5.2 (decreasing $f_0$).  False discovery proportions (A) and true positive rates (B) were averaged  over 100 simulation runs. Error bars (A) represent the  95\% CIs and the dashed horizontal line indicates the target FDR level of 0.05. FDRregT and FDRregE represent the FDRreg method using the theoretical and empirical null respectively.}

\label{fig:sim:13}
\end{figure}

\section{Numerical results for the mixed strategy}
Figure \ref{fig:sim:mix} provides some numerical results for the mixed strategy discussed in Section \ref{sec:dis}.

\begin{figure}
\centering
\includegraphics[width=0.9\textwidth]{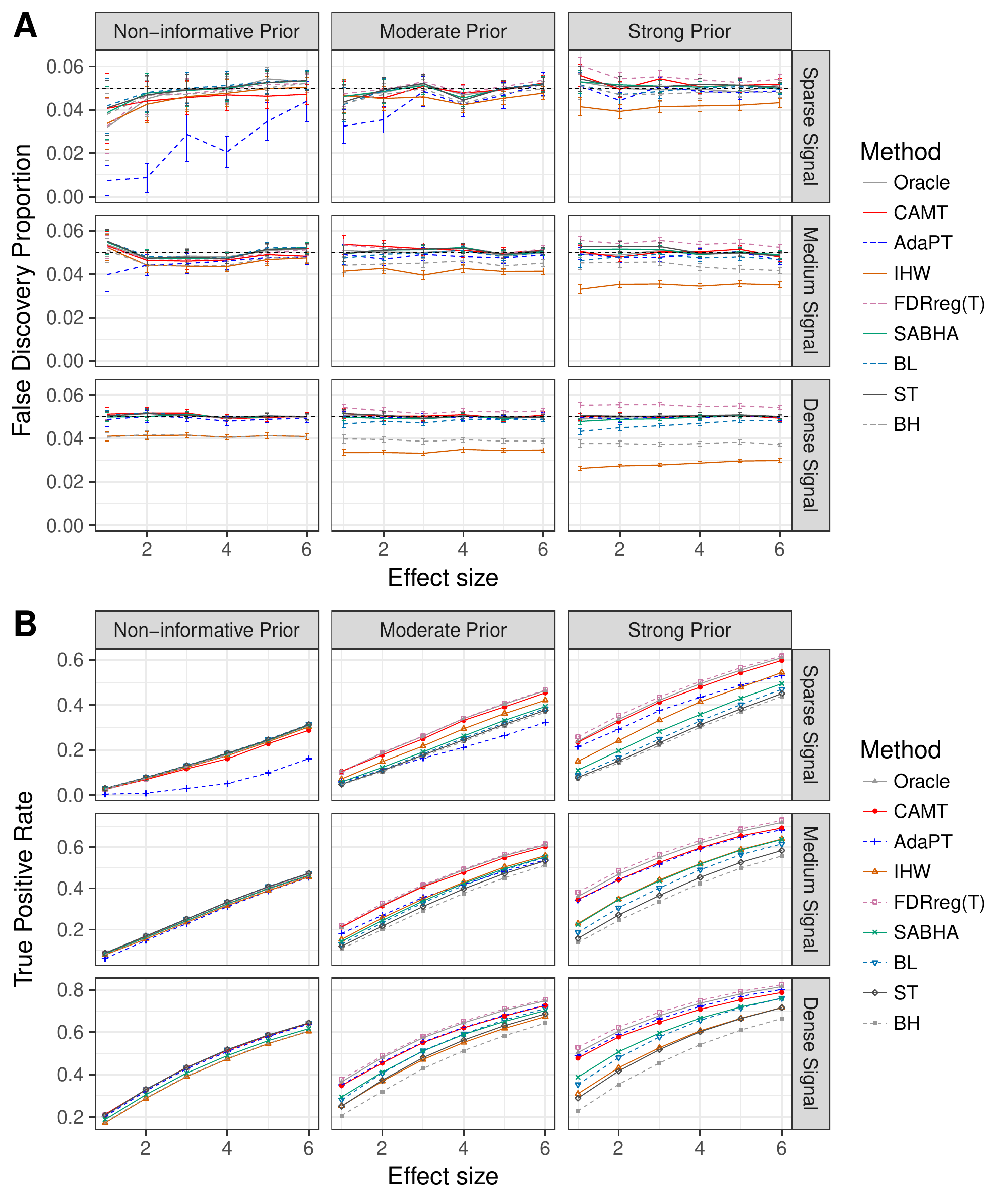}
\caption[Performance comparison under S0]{Performance comparison under the basic setting (S0). CAMT used the mixed strategy discussed in in Section \ref{sec:dis}. False discovery proportions (A) and true positive rates (B) were averaged  over 100 simulation runs. Error bars (A) represent the  95\% CIs and the dashed horizontal line indicates the target FDR level of 0.05. }
\label{fig:sim:mix}
\end{figure}


\begin{thebibliography}{99}
\bibitem{AC}
\par\noindent\hangindent2.3em\hangafter 1
Arias-Castro, E., and Chen, S. (2017). Distribution-free multiple testing. {\it Electronic Journal of Statistics}, \textbf{11}, 1983-2001.

\bibitem{bc15}
\par\noindent\hangindent2.3em\hangafter 1
Barber, R. F., and Cand\`{e}s, E. J. (2015). Controlling the false discovery rate via knockoffs. {\it Annals of Statistics}, \textbf{43}, 2055-2085.




\bibitem{bh95}
\par\noindent\hangindent2.3em\hangafter 1
Benjamini, Y., and Hochberg, Y. (1995). Controlling the false
discovery rate: a practical and powerful approach to multiple
testing. {\it Journal of the Royal Statistical Society, Series B},
\textbf{57}, 289-300.


\bibitem{bottomly}
\par\noindent\hangindent2.3em\hangafter 1
Bottomly, D., Walter, N.A., Hunter, J.E., Darakjian, P., Kawane, S., Buck, K.J., Searles, R.P., Mooney, M., McWeeney, S.K., and Hitzemann, R. (2011). Evaluating gene expression in C57BL/6J and DBA/2J mouse striatum using RNA-Seq and microarrays. {\it PloS one}, \textbf{6}, p.e17820.


\bibitem{by}
\par\noindent\hangindent2.3em\hangafter 1
Brooks, A. N., Yang, L., Duff, M. O., Hansen, K. D., Park, J. W., Dudoit, S., Brenner,  S. E.,
and Graveley, B. R. (2011). Conservation of an RNA regulatory map between Drosophila and mammals. {\it Genome Research}, \textbf{21}, 193-202.











\bibitem{cai09}
\par\noindent\hangindent2.3em\hangafter 1
Cai, T. T., and Sun, W. (2009). Simultaneous testing of grouped
hypotheses: finding needles in multiple haystacks. {\it Journal of
the American Statistical Association}, \textbf{104}, 1467--1481.




\bibitem{de}
\par\noindent\hangindent2.3em\hangafter 1
Dephoure, N., and Gygi, S. P. (2012). Hyperplexing: a method for higher-order multiplexed quantitative proteomics provides a map of the dynamic response to rapamycin in yeast.
{\it Science Signaling}, \textbf{5}, rs2-rs2.


\bibitem{Efron}
\par\noindent\hangindent2.3em\hangafter 1
Efron, B. (2004). Local false discovery rate. Technical report, Stanford University, Dept. of Statistics.


\bibitem{genovese2006}
\par\noindent\hangindent2.3em\hangafter 1
Genovese, C. R., Roeder, K., and Wasserman, L. (2006).
False discovery control with $p$-value weighting.
{\it Biometrika}, \textbf{93}, 509-524.



\bibitem{devlin99}
\par\noindent\hangindent2.3em\hangafter 1
Devlin, B., and Roeder, K.  (1999). Genomic control for association studies. {\it Biometrics}, \textbf{55}, 997-1004.

\bibitem{him}
\par\noindent\hangindent2.3em\hangafter 1
Himes, B.E., Jiang, X., Wagner, P., Hu, R., Wang, Q., Klanderman, B., Whitaker, R. M., Duan, Q., Lasky-Su, J., Nikolos, C., and Jester, W. (2014). RNA-Seq transcriptome profiling identifies CRISPLD2 as a glucocorticoid responsive gene that modulates cytokine function in airway smooth muscle cells. {\it PloS one}, \textbf{9}, p.e99625.


\bibitem{hu2010}
\par\noindent\hangindent2.3em\hangafter 1
Hu, J. X., Zhao, H., and Zhou, H. H. (2010).
False discovery rate control with groups.
{\it Journal of the American Statistical Association}, \textbf{105}, 1215-1227.


\bibitem{Ignatiadis2016}
\par\noindent\hangindent2.3em\hangafter 1
Ignatiadis, N., Klaus, B., Zaugg, J. B., and Huber, W. (2016).
Data-driven hypothesis weighting increases detection power in genome-scale multiple testing.
\textit{Nature Methods}, \textbf{13}, 577--580.



\bibitem{Leek2007}
\par\noindent\hangindent2.3em\hangafter 1
Leek, J. T.,  and Storey, J. D. (2007). Capturing heterogeneity in gene expression studies by surrogate variable analysis.
\textit{ PLoS genetics}, \textbf{3}, e161.



\bibitem{lf16}
\par\noindent\hangindent2.3em\hangafter 1
Lei, L., and Fithian, W. (2018). AdaPT: An interactive procedure for multiple testing with side information. \textit{Journal of the Royal Statistical Society, Series B}, to appear.


\bibitem{lb17}
\par\noindent\hangindent2.3em\hangafter 1
Li, A., and Barber, R. F. (2017). Multiple testing with the
structure adaptive Benjamini-Hochberg algorithm.
arXiv:1606.07926.





\bibitem{McDonald2018}
\par\noindent\hangindent2.3em\hangafter 1
McDonald, D., Hyde, E., Debelius, J.W., Morton, J.T.,  Gonzalez, A., Ackermann, G.,    et al. (2018).
American Gut: an open platform for citizen science microbiome research.
\textit{mSystems}, \textbf{3}, e00031--18.





\bibitem{pd}
\par\noindent\hangindent2.3em\hangafter 1
P\"{o}tscher, B. M., and Prucha, I. R. (1989). A uniform law of large numbers for dependent and heterogeneous data processes. {\it Econometrica}, 675-683.


\bibitem{Robertson2005}
\par\noindent\hangindent2.3em\hangafter 1
Robertson, K. D. (2005). DNA methylation and human disease.
\textit{Nature Reviews Genetics}, \textbf{6}, 597.



\bibitem{scott2015}
\par\noindent\hangindent2.3em\hangafter 1
Scott, J. G., Kelly, R. C., Smith, M. A., Zhou, P., and Kass, R. E.
(2015). False discovery rate regression: an application to neural
synchrony detection in primary visual cortex. \textit{Journal of the
American Statistical Association}, \textbf{110}, 459--471.


\bibitem{storey2002}
\par\noindent\hangindent2.3em\hangafter 1
Storey, J. D. (2002). A direct approach to false discovery rates.
\textit{Journal of the Royal Statistical Society, Series B}, \textbf{64}, 479--498.


\bibitem{sts2004}
\par\noindent\hangindent2.3em\hangafter 1
Storey, J. D., Taylor, J. E., and Siegmund, D. (2004). Strong
control, conservative point estimation and simultaneous conservative
consistency of false discovery rates: a unified approach. \textit{Journal of the Royal Statistical Society, Series B}, \textbf{66}, 187--205.



\bibitem{sunspatial15}
\par\noindent\hangindent2.3em\hangafter 1
Sun, W., Reich, B. J., Cai, T. T., Guindani, M., and Schwartzman, A.
(2015). False discovery control in large-scale multiple testing.
\textit{Journal of the Royal Statistical Society, Series B},
\textbf{77}, 59--83.


\bibitem{tansey2015}
\par\noindent\hangindent2.3em\hangafter 1
Tansey, W., Koyejo, O., Poldrack, R. A., and Scott, J. G. (2017).
False discovery rate smootinng. arXiv:1411.6144.


\bibitem{white}
\par\noindent\hangindent2.3em\hangafter 1
White, H. (1982). Maximum likelihood estimation of misspecified Models. {\it Econometrica}, \textbf{50}, 1-25.



\bibitem{Wijnands2017}
\par\noindent\hangindent2.3em\hangafter 1
Wijnands, K.P., Chen, J., Liang, L., Verbiest, M.M., Lin, X., Helbing, W.A., et al. (2016).
Genome-wide methylation analysis identifies novel CpG loci for perimembranous ventricular septal defects in human.
\textit{Epigenomics}, \textbf{9}, 241--251.





\end{thebibliography}

\begin{thebibliography}{99}

\bibitem{bc16}
\par\noindent\hangindent2.3em\hangafter 1
Barber, R. F., and Cand\`{e}s, E. J. (2016). A knockoff filter for high-dimensional selective inference. arXiv:1602.03574.


\bibitem{pd}
\par\noindent\hangindent2.3em\hangafter 1
P\"{o}tscher, B. M., and Prucha, I. R. (1989). A uniform law of large numbers for dependent and heterogeneous data processes. {\it Econometrica}, 675-683.


\bibitem{path}
\par\noindent\hangindent2.3em\hangafter 1
Resnick, S. I. (2005). {\it A probability path}. Springer Science \& Business Media.


\bibitem{white}
\par\noindent\hangindent2.3em\hangafter 1
White, H. (1982). Maximum likelihood estimation of misspecified models. {\it Econometrica}, \textbf{50}, 1-25.


\end{thebibliography}
\end{document}